\documentclass[11pt]{article}

\usepackage[margin=1in]{geometry}
\usepackage{graphicx}
\usepackage[pdftex,colorlinks=true,linkcolor=blue,citecolor=blue,urlcolor=black]{hyperref}
\usepackage{amsmath, amsthm, amssymb}
\usepackage{subfigure}
\usepackage{comment}
\usepackage{url}
\usepackage{pdflscape}
\usepackage[ruled,lined,linesnumbered]{algorithm2e}
\usepackage{tikz}
\usepackage{enumerate}

\newcommand{\R}{\mathbb{R}}

\newcommand{\C}{\mathbb{C}}

\newcommand{\ket}[1]{| #1 \rangle}
\newcommand{\bra}[1]{\langle #1|}

\newcommand{\proj}[1]{| #1 \rangle \langle #1 |}
\newcommand{\bracket}[3]{\langle #1|#2|#3 \rangle}
\newcommand{\sm}[1]{\left( \begin{smallmatrix} #1 \end{smallmatrix} \right)}

\DeclareMathOperator{\poly}{poly}

\DeclareMathOperator{\sgn}{sgn}
\DeclareMathOperator{\tr}{tr}

\DeclareMathOperator{\spann}{span}
\DeclareMathOperator{\spec}{spec}
\DeclareMathOperator{\diag}{diag}

\newcommand{\heff}{H_{\operatorname{eff}}}
\newcommand{\helse}{H_{\operatorname{else}}}
\renewcommand{\Re}{\operatorname{Re}}

\newcommand{\ptime}{\mathsf{P}}

\newcommand{\nptime}{\mathsf{NP}}

\newcommand{\bqptime}{\mathsf{BQP}}

\newcommand{\arthurmerlin}{\mathsf{AM}}
\newcommand{\merlinarthur}{\mathsf{MA}}
\newcommand{\qma}{\mathsf{QMA}}
\newcommand{\stoqma}{\mathsf{StoqMA}}
\newcommand{\tim}{\mathsf{TIM}}

\newcommand{\sham}{{\sc $\mathcal{S}$-Hamiltonian}}
\newcommand{\shamlf}{{\sc $\mathcal{S}$-Hamiltonian with local terms}}

\newcommand{\be}{\begin{equation}}
\newcommand{\ee}{\end{equation}}
\newcommand{\bea}{\begin{eqnarray}}
\newcommand{\eea}{\end{eqnarray}}
\newcommand{\bes}{\begin{equation*}}
\newcommand{\ees}{\end{equation*}}

\makeatletter
\newtheorem*{rep@theorem}{\rep@title}
\newcommand{\newreptheorem}[2]{%
\newenvironment{rep#1}[1]{%
 \def\rep@title{#2 \ref{##1} (restated)}%
 \begin{rep@theorem}}%
 {\end{rep@theorem}}}
\makeatother

\newtheorem{thm}{Theorem}
\newtheorem*{thm*}{Theorem}
\newtheorem{cor}[thm]{Corollary}

\newtheorem{lem}[thm]{Lemma}
\newtheorem*{lem*}{Lemma}
\newtheorem{prop}[thm]{Proposition}
\newtheorem{dfn}[thm]{Definition}

\newtheorem{obs}[thm]{Observation}

\newtheorem{prob}[thm]{Problem}

\newreptheorem{thm}{Theorem}
\newreptheorem{lem}{Lemma}

\begin{document}

\title{Complexity classification of local Hamiltonian problems\footnote{An extended abstract of this work appeared in the proceedings of FOCS 2014.}}
\author{Toby Cubitt\thanks{Department of Computer Science, University College London, UK; {\tt t.cubitt@ucl.ac.uk}.}
\and
Ashley Montanaro\thanks{School of Mathematics, University of Bristol, UK; {\tt ashley.montanaro@bristol.ac.uk}. Some of this work was done while the authors were at the University of Cambridge.}}
\maketitle

\begin{abstract}
The calculation of ground-state energies of physical systems can be formalised as the {\sc $k$\nobreakdash-local Hamiltonian} problem, which is a natural quantum analogue of classical constraint satisfaction problems. One way of making the problem more physically meaningful is to restrict the Hamiltonian in question by picking its terms from a fixed set $\mathcal{S}$, and scaling them by arbitrary weights. Examples of such special cases are the Heisenberg and Ising models from condensed-matter physics.

In this work we characterise the complexity of this problem for all 2\nobreakdash-local qubit Hamiltonians. Depending on the subset $\mathcal{S}$, the problem falls into one of the following categories: in $\ptime$; $\nptime$-complete; polynomial-time equivalent to the Ising model with transverse magnetic fields; or $\qma$-complete. The third of these classes has been shown to be $\stoqma$-complete by Bravyi and Hastings. The characterisation holds even if $\mathcal{S}$ does not contain any 1\nobreakdash-local terms; for example, we prove for the first time $\qma$-completeness of the Heisenberg and XY interactions in this setting. If $\mathcal{S}$ is assumed to contain all 1\nobreakdash-local terms, which is the setting considered by previous work, we have a characterisation that goes beyond 2\nobreakdash-local interactions: for any constant $k$, all $k$\nobreakdash-local qubit Hamiltonians whose terms are picked from a fixed set $\mathcal{S}$ correspond to problems either in $\ptime$; polynomial-time equivalent to the Ising model with transverse magnetic fields; or $\qma$-complete.

These results are a quantum analogue of the maximisation variant of Schaefer's dichotomy theorem for boolean constraint satisfaction problems.
\end{abstract}



\section{Introduction}


Constraint satisfaction problems (CSPs) are ubiquitous in computer science and have been intensively studied since the early days of complexity theory.
A beautiful and surprising result in this area is the dichotomy theorem of Schaefer~\cite{schaefer78}, which {\em completely} classifies the complexity of boolean constraint satisfaction problems of a certain form. These problems can all be considered special cases of a general problem $\mathcal{S}$-CSP, where $\mathcal{S}$ is a set of constraints, each of which is a boolean function on a fixed number of bits. An instance of the problem is described by a sequence of these constraints, applied to different subsets of input bits. The task is to determine whether all the constraints can be simultaneously satisfied. For example, the 3-SAT problem fits into this class: here the constraints are disjunctions of up to 3 input bits, or their negations. Schaefer's result states that if $\mathcal{S}$ is one of a particular family of types of constraints, $\mathcal{S}$-CSP is in $\ptime$; otherwise, $\mathcal{S}$-CSP is $\nptime$-complete. This result is particularly remarkable given Ladner's theorem~\cite{ladner75} that, assuming $\ptime \neq \nptime$, there must be an infinite hierarchy of complexity classes between $\ptime$ and $\nptime$.

Schaefer's dichotomy theorem has subsequently been generalised and sharpened in a number of directions. In particular, Creignou~\cite{creignou95} and Khanna, Sudan and Williamson~\cite{khanna97} have completely characterised the complexity of the maximisation problem {\sc $k$-Max-CSP} for boolean constraints. Here we are again given a system of constraints, but the goal is to maximise the number of constraints we can satisfy. An example problem of this kind is MAX-CUT. A recent monograph of Creignou, Khanna and Sudan~\cite{creignou01} has much more on this subject.

A natural quantum generalisation of constraint satisfaction problems is provided by the {\sc $k$\nobreakdash-local Hamiltonian} problem~\cite{kitaev02}.
A $k$\nobreakdash-local Hamiltonian is a Hermitian matrix $H$ on the space of $n$ qubits which can be written as
\[ H = \sum_i H^{(i)}, \]
where each $H^{(i)}$ acts non-trivially on at most $k$ qubits, i.e.\ it is of the form $H^{(i)} = h\otimes I$ where $h$ is a Hamiltonian on at most $k$ qubits.

\begin{dfn}[{\sc $k$\nobreakdash-local Hamiltonian}]
\label{dfn:klham}
The (promise) problem {\sc $k$\nobreakdash-local Hamiltonian} is defined as follows. We are given a $k$\nobreakdash-local Hamiltonian $H = \sum_{i=1}^m H^{(i)}$ on $n$ qubits with $m = \poly(n)$. Each $H^{(i)}$ satisfies $\|H^{(i)}\| = \poly(n)$ and its entries are specified by $\poly(n)$ bits. We are also given two rational numbers $a<b$ of $\poly(n)$ digits such that $b-a \ge 1/\poly(n)$, and promised that the smallest eigenvalue of $H$ is either at most $a$, or at least $b$. Our task is to determine which of these two possibilities is the case.
\end{dfn}


{\sc $k$\nobreakdash-local Hamiltonian} is a direct generalisation of {\sc $k$-Max-CSP}; the classical problem is the special case where each matrix $H^{(i)}$ is diagonal in the computational basis and only contains 0's and 1's.  Just as {\sc $k$-Max-CSP} is $\nptime$-complete for $k \ge 2$, {\sc $k$\nobreakdash-local Hamiltonian} is $\qma$-complete for $k \ge 2$~\cite{kempe06}, where $\qma$ (quantum Merlin-Arthur) is the quantum analogue of $\nptime$~\cite{kitaev02}. (We give a formal definition of this class in Appendix~\ref{app:classes}.) If a problem is $\qma$-complete, this is good evidence that there is unlikely to be a polynomial-time algorithm (whether classical or quantum) to solve it.


As well as the intrinsic mathematical interest of this noncommutative generalisation of constraint satisfaction problems, a major motivation for this area is applications to physics. Indeed, the classical connection between constraint satisfaction and physics goes back at least as far as Barahona's work proving $\nptime$-hardness of cases of the Ising model~\cite{barahona82}. One of the most important themes in condensed-matter physics is calculating the ground-state energies of physical systems\footnote{In practice, one might often actually like to determine some more complicated property of the ground state; however, calculating the energy is a reasonable starting point.}; this is essentially an instance of {\sc $k$\nobreakdash-local Hamiltonian}. This connection to physics motivates the study of the $\qma$-hardness (or otherwise) of {\sc $k$\nobreakdash-local Hamiltonian} with restricted types of interactions, with the aim being to prove $\qma$-hardness of problems of more direct physical interest, rather than the somewhat unnatural interactions that may occur in the general {\sc $k$\nobreakdash-local Hamiltonian} problem. This is the quantum analogue of the classical programme of proving $\nptime$-hardness of constraint satisfaction problems where the constraints are picked from a restricted set $\mathcal{S}$. One can also consider {\sc $k$\nobreakdash-local Hamiltonian} with restrictions on the interaction geometry (for example, taking all interactions to be 2\nobreakdash-local on a planar lattice).


In particular, it is known that {\sc 2\nobreakdash-local Hamiltonian} remains $\qma$-complete if:
\begin{itemize}
\item the Hamiltonian $H$ is of the Heisenberg form with arbitrary local magnetic fields,
\[ H = \sum_{(i,j) \in E} X_i X_j + Y_i Y_j + Z_i Z_j + \sum_k \alpha_k X_k + \beta_k Y_k + \gamma_k Z_k, \]
where $\alpha_k$, $\beta_k$, $\gamma_k$ are arbitrary coefficients and $E$ is the set of edges of a 2-dimensional square lattice~\cite{oliveira08,schuch09};

\item the Hamiltonian $H$ is of the form~\cite{biamonte08}
\[ H = \sum_{i<j} J_{ij} X_i X_j + K_{ij} Z_i Z_j + \sum_k \alpha_k X_k + \beta_k Z_k, \]
or
\[ H= \sum_{i<j} J_{ij} X_i Z_j + K_{ij} Z_i X_j + \sum_k \alpha_k X_k + \beta_k Z_k, \]
\end{itemize}
where $J_{ij}$, $K_{ij}$, $\alpha_k$, $\beta_k$ are arbitrary coefficients. These results determine the complexity of various special cases of the following general problem, which we call \sham.

\begin{dfn}[\sham]
Let $\mathcal{S}$ be a fixed finite set of Hermitian matrices such that each matrix in $\mathcal{S}$ acts on at most $k=O(1)$ qubits, and acts nontrivially on all of these qubits. The {\sc $\mathcal{S}$-Hamiltonian} problem is the special case of {\sc $k$\nobreakdash-local Hamiltonian} where, for each $i$, there exists $\alpha_i \in \R$ such that $\alpha_i H^{(i)} \in \mathcal{S}$. That is, the overall Hamiltonian $H$ is specified by a sum of matrices $H^{(i)}$, each of which acts non-trivially on at most $k$ qubits, and whose non-trivial part is proportional to a matrix picked from $\mathcal{S}$.
\end{dfn}

We then have the following general question:

\begin{prob}
\label{prob:main}
Given $\mathcal{S}$, characterise the computational complexity of {\sc $\mathcal{S}$-Hamiltonian}.
\end{prob}

We will essentially completely resolve this question in the case where every matrix in $\mathcal{S}$ acts on at most 2 qubits. Before we state our results, we observe the following important points about this problem:

\begin{itemize}
\item In general, we assume that, given a set of interactions $\mathcal{S}$, we are allowed to produce an overall Hamiltonian by applying each interaction $M \in \mathcal{S}$ scaled by an arbitrary real weight, which can be either positive or negative. This contrasts with classical constraint satisfaction problems, where usually weights are restricted to be positive.

\item We assume that we are allowed to apply the interactions in $\mathcal{S}$ across any choice of subsets of the qubits. That is, the interaction pattern is not constrained by any spatial locality, planarity or symmetry considerations. (Note that classically, it is common to allow one constraint in a CSP to take as input multiple copies of the same variable. But this seems less meaningful from a physical perspective, and for quantum CSPs it's not clear if it makes sense at all. So we do not allow it here. Even classically, this distinctness requirement can make it more difficult to prove hardness for families of CSPs~\cite{khanna97}.)

\item Some of the interactions in $\mathcal{S}$ could be non-symmetric under permutation of the qubits on which they act; for example, it could make a difference whether we apply $M \in \mathcal{S}$ across qubits $(1,2)$ or qubits $(2,1)$. We assume that we are allowed to apply such interactions to any permutation of the qubits.

\item We can always assume without loss of generality that the identity matrix $I \in \mathcal{S}$, as adding an arbitrarily weighted identity term (energy shift) does not change the hardness of the problem.
\end{itemize}

Making these assumptions will allow us to give a precise classification of the complexity of \sham; the price paid is that the problem instances considered are potentially less physically meaningful (for example, containing terms with polynomially large weights, with both positive and negative signs, and with interactions across large distances). Finding a full characterisation of \sham\ with additional restrictions on the form of the Hamiltonians considered seems to be a very challenging task. However, sometimes (see Section~\ref{sec:results} below) we are nevertheless able to classify the complexity of \sham\ even when restricted to more physically realistic Hamiltonians.

A number of interesting special cases of {\sc $k$\nobreakdash-local Hamiltonian} which do not exactly fit into the {\sc $\mathcal{S}$-Hamiltonian} framework have also been studied. The problem remains $\qma$-complete for bosonic~\cite{wei10} and fermionic~\cite{liu07a} systems. In another direction, it has been shown by Bravyi et al.~\cite{bravyi06} that {\sc $k$\nobreakdash-local Hamiltonian} is in the complexity class $\arthurmerlin$ if the Hamiltonian is restricted to be {\em stoquastic}. A stoquastic Hamiltonian has all off-diagonal entries real and non-positive in the computational basis. Such Hamiltonians are of particular interest as they occur in a wide variety of physical systems, and also in the quantum adiabatic algorithm for SAT~\cite{farhi00} and certain claimed implementations of quantum computation~\cite{johnson11}. Monte-Carlo calculations for stoquastic Hamiltonians do not suffer from the sign problem~\cite{bravyi14}, so the intuition from physics is that such Hamiltonians are easier to simulate classically than the general case. Indeed, as $\arthurmerlin$ is in the polynomial hierarchy, it is considered unlikely that {\sc $k$\nobreakdash-local Hamiltonian} with stoquastic Hamiltonians is $\qma$-complete. This result was subsequently sharpened by Bravyi, Bessen and Terhal~\cite{bravyi06a}, who showed that this problem is $\stoqma$-complete, where $\stoqma$ is a complexity class which sits between $\merlinarthur$ and $\arthurmerlin$; we include a formal definition in Appendix~\ref{app:classes}. On the other hand, approximating the {\em highest} eigenvalue of a stoquastic Hamiltonian is $\qma$-complete~\cite{jordan10}. Bravyi and Vyalyi~\cite{bravyi05} proved that {\sc $k$\nobreakdash-local Hamiltonian} is in $\nptime$ for 2\nobreakdash-local Hamiltonians with commuting terms, and this has been extended recently by Hastings~\cite{hastings11} to further classes of commuting Hamiltonians.


\subsection{Statement of results}
\label{sec:results}

We begin by considering a special case of the problem, which we call \shamlf\ and is defined as follows.

\begin{dfn}
\shamlf\ is the special case of \sham\ where $\mathcal{S}$ is assumed to contain the Pauli matrices $\{X,Y,Z\}$.
\end{dfn}

As the Pauli matrices span the space of 1-qubit Hermitian matrices, this is equivalent to assuming that $\mathcal{S}$ contains all 1-qubit Hermitian matrices. That is, in the \shamlf\ problem we are given access to all 1\nobreakdash-local terms for free: the overall Hamiltonian is formed by taking a sum of terms from $\mathcal{S}$, each with an arbitrary positive or negative weight, then adding arbitrary 1\nobreakdash-local terms. For any $\mathcal{S}$, \shamlf\ is at least as difficult as \sham, because it is a generalisation. It is therefore easier to prove $\qma$-hardness of cases of \shamlf. All previous proofs of $\qma$-hardness of special cases of {\sc $k$\nobreakdash-local Hamiltonian} which we are aware of~\cite{kempe06,biamonte08,oliveira08,schuch09} actually prove $\qma$-hardness of \shamlf\ for various sets $\mathcal{S}$. Here we are able to characterise the complexity of this problem when $\mathcal{S}$ contains arbitrary matrices on up to $k$ qubits, for arbitrary $k=O(1)$.

We first need to define a notion of local diagonalisation. Let $M$ be a $k$-qubit Hermitian matrix. We say that $U \in SU(2)$ {\em locally diagonalises} $M$ if $U^{\otimes k} M (U^{\dag})^{\otimes k}$ is diagonal. We say that $U$ locally diagonalises $\mathcal{S}$ if $U$ locally diagonalises $M$ for all $M \in \mathcal{S}$. Note that matrices in $\mathcal{S}$ may act on different numbers of qubits, so can be of different sizes. To gain some intuition for why this notion should be relevant, observe that if we have an instance $H$ of \sham\ which is a sum of many terms acting on overlapping sets of qubits, and if there is a $U$ that locally diagonalises $\mathcal{S}$, then all the terms in $H$ can be diagonalised simultaneously by $U^{\otimes n}$. On the other hand, if all the matrices in $\mathcal{S}$ are simultaneously diagonalisable but are not locally diagonalisable, it is not necessarily obvious how to diagonalise $H$.

We also need a notion of what it means to remove the 1\nobreakdash-local part of a $k$-local Hamiltonian. We will only need this for qubits, so we restrict the definition to this case for simplicity. Since the Pauli matrices together with the identity matrix form a complete basis for $2\times 2$ matrices, any $k$-qubit matrix $M$ can be decomposed as a weighted sum of tensor products of Pauli and identity matrices. ``Removing the 1\nobreakdash-local part'' of $M$ means the matrix produced by deleting all 1\nobreakdash-local terms from its Pauli decomposition.

We are now ready to state our first main result.

\begin{thm}
\label{thm:shamlf}
Let $\mathcal{S}$ be an arbitrary fixed subset of Hermitian matrices on at most $k$ qubits, where $k=O(1)$. Let  $\mathcal{S}'$ be the subset formed by removing the 1\nobreakdash-local part from each element of $\mathcal{S}$, and then deleting all 0\nobreakdash-local matrices from the resulting set. Then:
\begin{itemize}
\item If $\mathcal{S}'$ is empty, \shamlf\ is in $\ptime$;
\item Otherwise, if there exists $U \in SU(2)$ such that $U$ locally diagonalises $\mathcal{S}'$, then \shamlf\ is $\stoqma$-complete.
\item Otherwise, \shamlf\ is $\qma$-complete. If every matrix in $\mathcal{S}'$ acts on 2 qubits, this holds even if we insist that the 2-qubit interactions in the final Hamiltonian are restricted to the edges of a 2d square lattice and all have equal weight.
\end{itemize}
\end{thm}

As the existence of a unitary $U$ required for the second condition to hold can be checked by solving a system of polynomial equations, $\mathcal{S}$ is finite, and every matrix in $\mathcal{S}$ is of finite size, it is clear that this classification can be performed in constant time independent of the problem instance. We give two alternative characterisations of diagonalisability by local unitaries in Appendix~\ref{app:dlu}, which imply more efficient algorithms. The alert reader may wonder why there is no $\nptime$-complete class in the above characterisation; this is because of the free 1\nobreakdash-local terms allowed.

The even more alert reader may also wonder at the presence of $\stoqma$ in the classification, given that we allow terms in $\mathcal{S}$ to be used with arbitrary weights with both signs. This implies that one can always produce a non-stoquastic Hamiltonian from any set $\mathcal{S}$ containing a non-diagonal matrix, even if all the elements of $\mathcal{S}$ have real non-positive off-diagonal entries, thus suggesting that stoquasticity is not a meaningful constraint in our setting. The $\stoqma$-complete class of interactions can in fact be understood as those special cases of \shamlf\ which are polynomial-time equivalent to the problem of approximating the lowest eigenvalue of a Hamiltonian in the general Ising model with transverse magnetic fields. This model (whose name we will shorten to ``transverse Ising model'') describes Hamiltonians of the form\footnote{This model is sometimes defined to additionally include single-qubit $Z$ terms~\cite{bravyi14b}. However, these can be generated by $ZZ$ interactions with an ancilla, as shown at the end of Section~\ref{sec:zzlf}.}
\be \label{eq:tim} H = \sum_{i < j} \alpha_{ij} Z_i Z_j + \sum_k \beta_k X_k. \ee
Such Hamiltonians have been much studied in mathematical physics and in particular occur in the quantum adiabatic algorithm for solving optimisation problems~\cite{farhi00}. In our terminology, the problem of determining the ground-state energy of Hamiltonians in the transverse Ising model up to inverse-polynomial precision is {\sc $\{ZZ,X\}$-Hamiltonian}.


Initial versions of this work characterised the second class in Theorem \ref{thm:shamlf} as problems interreducible with {\sc $\{ZZ,X\}$-Hamiltonian}, and introduced a corresponding new complexity class $\tim$. Bravyi and Hastings~\cite{bravyi14b} subsequently proved the following result:

\begin{thm}[Bravyi and Hastings~\cite{bravyi14b}]
\label{thm:bh}
{\sc $\{ZZ,X\}$-Hamiltonian} is $\stoqma$-complete.
\end{thm}

In other words, $\tim = \stoqma$. To gain some intuition for this class of Hamiltonians, observe that fixing $\beta_k = 0$ suffices to show that {\sc $\{ZZ,X\}$-Hamiltonian} is $\nptime$-hard, by the $\nptime$-hardness of the general Ising model. On the other hand, {\sc $\{ZZ,X\}$-Hamiltonian} can be seen to be in $\stoqma$ by the following argument. By conjugating any Hamiltonian $H$ of the form (\ref{eq:tim}) by local $Z$ operations on each qubit $k$ such that $\beta_k > 0$, which maps $X \mapsto -X$ and does not change the eigenvalues of $H$, $\beta_k$ can be assumed to be non-positive for all $k$. The resulting Hamiltonian has all off-diagonal entries non-positive, or in other words is stoquastic~\cite{bravyi06}, so $\text{\sc $\{ZZ,X\}$-Hamiltonian} \in \stoqma$. As $\stoqma$ is contained within the polynomial hierarchy, and in particular within the class $\arthurmerlin$~\cite{bravyi06}, it is unlikely that $\stoqma = \qma$. Thus, for sets $\mathcal{S}$ which fall into this second class, \shamlf\ is unlikely to be $\qma$-complete.




We can go further than Theorem~\ref{thm:shamlf}, and consider a setting where we do not necessarily have access to all (or any) 1-qubit matrices. In this case, we can still completely characterise the complexity of \sham\ for all sets $\mathcal{S}$ of 2-qubit Hermitian matrices, with a slightly more complicated classification.

\begin{thm}
\label{thm:sham}
Let $\mathcal{S}$ be an arbitrary fixed subset of Hermitian matrices on at most 2 qubits. Then:
\begin{itemize}
\item If every matrix in $\mathcal{S}$ is 1\nobreakdash-local, \sham\ is in $\ptime$;
\item Otherwise, if there exists $U \in SU(2)$ such that $U$ locally diagonalises $\mathcal{S}$, then \sham\ is $\nptime$-complete;
\item Otherwise, if there exists $U \in SU(2)$ such that, for each 2-qubit matrix $H_i \in \mathcal{S}$, $U^{\otimes 2} H_i (U^{\dag})^{\otimes 2} = \alpha_i Z^{\otimes 2} + A_i\otimes I + I\otimes B_i$, where $\alpha_i \in \R$ and $A_i$, $B_i$ are arbitrary single-qubit Hermitian matrices, then \sham\ is $\stoqma$-complete;
\item Otherwise, \sham\ is $\qma$-complete.
\end{itemize}
\end{thm}

For 2\nobreakdash-local Hamiltonians, Theorem~\ref{thm:sham} is essentially a generalisation of Theorem~\ref{thm:shamlf}, except that in Theorem \ref{thm:shamlf} the $\qma$-complete cases are $\qma$-complete even on a square lattice. However, we state and prove them as two separate results as the proof techniques used are quite different.


In a sense, our result completely solves Kitaev's original qubit local Hamiltonian problem~\cite{kitaev02} for the case of two-body interactions (the most physically relevant case of the original qubit local-Hamiltonian problem). We highlight some interesting special cases, which are important models in mathematical physics.

\begin{itemize}
\item The {\bf general Heisenberg model} describes Hamiltonians of the following form:
\[ \sum_{i < j} \alpha_{ij} (X_i X_j + Y_i Y_j + Z_i Z_j). \]
In our terminology, this corresponds to {\sc $\{XX+YY+ZZ\}$-Hamiltonian}. By Theorem~\ref{thm:sham}, finding the ground-state energy of Hamiltonians in this model is $\qma$-complete. Prior to this work, this problem was not even known to be $\nptime$-hard. We stress that the $\alpha_{ij}$ coefficients are allowed to be independently positive or negative, while in physical systems one often restricts them to be either all positive (the {\em antiferromagnetic} case) or all negative (the {\em ferromagnetic} case); see Section~\ref{sec:realistic} for a further discussion of this point.

Schuch and Verstraete~\cite{schuch09} previously proved $\qma$-hardness of the Heisenberg model where arbitrary 1\nobreakdash-local terms are also allowed, or in other words $\qma$-hardness of {\sc $\{XX+YY+ZZ\}$-Hamiltonian with local terms}. The case where no local terms are allowed is particularly interesting because it displays a large amount of symmetry; indeed, the ground space of such a Hamiltonian on $n$ qubits must be invariant under conjugation by $U^{\otimes n}$ for arbitrary single-qubit unitaries $U$. Since the Heisenberg interaction is equivalent to projecting onto the two-qubit antisymmetric state (singlet), it can be viewed as a natural quantum generalisation of the MAX-CUT problem.

Theorem~\ref{thm:sham} also implies $\qma$-completeness of the XYZ model in condensed-matter physics, which corresponds to Hamiltonians of the form
\[ \sum_{i < j} \alpha_{ij} X_i X_j + \beta_{ij} Y_i Y_j + \gamma_{ij} Z_i Z_j, \]
and its special case, the XXZ model, with the constraint $\alpha_{ij} = \beta_{ij}$.

\item The {\bf general XY model} describes Hamiltonians of the following form:
\[ \sum_{i < j} \alpha_{ij} (X_i X_j + Y_i Y_j). \]
In the condensed-matter literature, this is often known as the XX model, and ``XY model'' is used for the more general case where the $X_i X_j$ and $Y_i Y_j$ terms can have different weights. When we refer to the XY model, we always mean the above special case. By Theorem~\ref{thm:sham}, calculating ground state energies in these models is $\qma$-complete. Biamonte and Love~\cite{biamonte08} previously proved $\qma$-completeness if arbitrary local $X$, $Y$ terms are allowed and the $X_i X_j$ and $Y_i Y_j$ terms can have different weights ($Y$ is relabelled to $Z$ in their work). It is worth remarking that, if the signs $\alpha_{ij}$ are restricted to be negative, the resulting Hamiltonian is stoquastic. Therefore, we see that $\qma$-completeness can be obtained from a simple special case of the stoquastic local Hamiltonian problem by allowing weights with varying signs.

%
%
\end{itemize}

Although the theory of $\qma$-completeness is now over a decade old~\cite{kitaev02}, the list of problems proven $\qma$-complete is still relatively short (see~\cite{bookatz14} for a recent review). One of the original motivations for Schaefer's dichotomy theorem~\cite{schaefer78} was to make $\nptime$-hardness proofs easier, by increasing the repertoire of $\nptime$-hard problems for use in reductions. We hope that our resolution of the complexity of \sham\ will be similarly useful to those wishing to prove $\qma$-hardness.


\subsubsection{Independent and subsequent work}

In independent work, Childs, Gosset and Webb~\cite{childs13} showed that the Bose-Hubbard model is $\qma$-complete. In proving this result, they showed that for Hamiltonians of the form
\[ \sum_{i \neq j, A_{ij}=1} X_i X_j + Y_i Y_j - \sum_{k,A_{kk}=1} Z_k, \]
where $A$ is the adjacency matrix of a graph, approximating the lowest eigenvalue restricted to a subspace with fixed expectation value of $Z^{\otimes n}$ (``magnetisation'') is $\qma$-complete. Their work thus showed that a variant of the {\sc $\{XX+YY,Z\}$-Hamiltonian} problem with an additional restriction to a subspace is $\qma$-complete, even if the non-zero coefficients of the terms are fixed to 1 (for $XX+YY$ terms) or $-1$ (for $Z$ terms). The same authors have subsequently shown that this even holds when there are no $Z$ terms in the Hamiltonian~\cite{childs15}. The problems studied by these authors contrast with \sham\ in that a problem instance is just a description of which qubits interact with which others. Here we have the additional freedom of choosing the weights of the interactions, but we do not allow any additional restriction to a specific subspace, so the two results are not directly comparable.

Recent work by Piddock and one of us~\cite{piddock15} has resolved some of the open questions from an earlier version of this paper, by showing that the antiferromagnetic XY and Heisenberg interactions are $\qma$-complete, and that many of the interactions proven $\qma$-complete here remain $\qma$-complete if restricted to a square lattice. For some interactions, these sign and geometry restrictions can be combined; for example, the antiferromagnetic XY interaction on a triangular lattice is $\qma$-complete.

Bravyi~\cite{bravyi14} has given a polynomial-time algorithm for approximating the ground-state energy of transverse Ising model Hamiltonians
\[ H = \sum_{i \neq j} \alpha_{ij} Z_i Z_j + \sum_k \beta_k X_k \]
in the ferromagnetic case where $\alpha_{ij} \le 0$ for all $i \neq j$. Finally, as mentioned above, Bravyi and Hastings~\cite{bravyi14b} sharpened the complexity classification in initial versions of this paper by proving that the {\sc 2\nobreakdash-local Hamiltonian} problem, restricted to transverse Ising model Hamiltonians, is actually $\stoqma$-complete.


\subsection{Proof techniques}

As is typical for ``dichotomy-type'' results, our classification theorems proceed by identifying some special cases which are easy, and then proving hardness of all other cases. All of our hardness results are based on reductions using gadgets (as used in e.g.~\cite{kempe06,oliveira08,schuch09,biamonte08}), rather than proving $\qma$-hardness directly using clock constructions or similar (as used in e.g.~\cite{kitaev02,kempe06}).

The basic idea is to approximately simulate some set of interactions $\mathcal{A}$, where {\sc $\mathcal{A}$-Hamiltonian} is $\qma$-hard, using some other set of interactions $\mathcal{B}$, thus proving $\qma$-hardness of {\sc $\mathcal{B}$-Hamiltonian}. We use two kinds of gadgets, both analysed using perturbation theory~\cite{kempe06}. This theory allows us to characterise the low-energy part of operators of the form $V + \Delta H$, where $V$ and $H$ are Hamiltonians and $\Delta = \poly(\|V\|)$ is a large coefficient. The simpler type of gadget consists of choosing a large enough constant $\Delta$ such that $V$ is effectively projected onto the ground space of $H$. This is the quantum analogue of the natural classical technique of forcing some input bits to be in a certain state by applying a heavily weighted constraint to them. A more complicated type of gadget does not have a classical analogue. Here we choose $H$ to be 1\nobreakdash-local, and by picking somewhat smaller $\Delta$, implement an effective 2\nobreakdash-local interaction which we did not have access to previously.

The \sham\ problem contains a daunting number of cases, so the first step of our proof is to reduce the Hamiltonians we consider to a normal form by conjugating by local unitaries, which does not change the eigenvalues. In addition, we can assume that every 2-qubit interaction $H \in \mathcal{S}$ is either symmetric or antisymmetric under interchange of the qubits on which it acts, via a symmetrisation argument. It turns out that, by conjugating by local unitaries, the 2\nobreakdash-local part of any given 2-qubit Hermitian matrix $H$ which is symmetric (resp.\ antisymmetric)  can be reduced to a matrix of the form $\alpha XX + \beta YY + \gamma ZZ$ (resp.\ $\alpha(XZ-ZX)$), which drastically reduces the number of cases we need to cover. If our set $\mathcal{S}$ contains more than one interaction, we need to be careful to apply the same local unitaries to all $H \in \mathcal{S}$.

In the case of \shamlf, the techniques used to prove $\qma$-hardness are then fairly standard (following previous work~\cite{kempe06,oliveira08,schuch09,biamonte08}). We use our access to arbitrary 1\nobreakdash-local terms to create perturbative gadgets which allow us to produce arbitrary interactions from interactions of the form $\alpha XX + \beta YY + \gamma ZZ$. In the case where $\mathcal{S}$ only contains interactions on 1 or 2 qubits, following this approach allows us to prove $\qma$-hardness even when all the 2-qubit interactions are equally weighted and are restricted to the edges of a 2d square lattice. We can also prove $\qma$-hardness for $k$-qubit interactions for $k>2$, which is based on using 1\nobreakdash-local interactions to ``cut out'' components of the $k$\nobreakdash-local interactions and produce 2-qubit interactions\footnote{It is not obvious how to achieve this without having access to 1\nobreakdash-local interactions, which is one reason why we were unable to achieve a full classification result for \sham\ for $k>2$.}. The $\stoqma$-complete class of interactions is characterised by showing that {\sc $\{ZZ\}$-Hamiltonian with local terms} reduces to {\sc $\{ZZ,X\}$-Hamiltonian} and using Theorem \ref{thm:bh}.

In the more general case of \sham\ it is more difficult to prove $\qma$-hardness, as the lack of access to 1\nobreakdash-local terms does not allow us to use the perturbative techniques of~\cite{oliveira08,schuch09}. In some cases, we are also hampered by the presence of symmetry. This is exemplified by the Heisenberg model $\mathcal{S} = \{XX+YY+ZZ\}$. As $H = XX + YY + ZZ$ is invariant under conjugation by $U^{\otimes 2}$ for any single-qubit unitary $U$, the same holds for the ground space of any Hamiltonian built only from $H$ terms, implying that it is hopeless to attempt to directly encode the ground state of a general Hamiltonian into a Heisenberg Hamiltonian. We therefore proceed using an encoding method where we associate a block of 3 physical qubits with a single logical qubit. This is inspired by related ideas in work on universality of the exchange interaction for quantum computation~\cite{kempe00}, but does not appear to follow from it directly. In order to make the encoding work, we use perturbation theory to effectively project onto a subspace which we can control within the 3-qubit space. An interesting aspect of the proof is that, in order to produce the correct interactions, we need as part of the construction to find a special case of the Heisenberg model with certain characteristics. The ground state $\ket{\psi}$ should be unique and globally entangled, and the correlation functions (quantities of the form $\bracket{\psi}{X_i X_j}{\psi}$) should be efficiently computable. Very few such cases exist, but luckily the previously studied Lieb-Mattis model~\cite{lieb62} has the properties we need. We believe that the ideas used in the proof of hardness of the general Heisenberg model (Section~\ref{sec:heisenberg}) are the most significant individual technical contributions in this work.

The other important special case with significant symmetry is the XY model $\mathcal{S} = \{XX + YY\}$. This can be dealt with using similar ideas, but the exactly solvable case used is somewhat simpler. These two cases are an integral part of proving hardness for more general interactions. It turns out that using a number of different encodings we can produce virtual interactions of either Heisenberg or XY type using almost any 2-qubit interaction with no 1\nobreakdash-local part, sufficing to prove $\qma$-hardness for these cases too. Finally, $\qma$-hardness of cases with 1\nobreakdash-local parts is proven by yet another gadget construction, this time one which removes the unwanted 1\nobreakdash-local terms.

In many of these cases, we needed to carry out fairly complicated eigenvalue-eigenvector calculations in order to prove that our gadgets work. These calculations were performed using a computer algebra package. However, once they are found, verifying that the eigenvectors and eigenvalues are correct can easily be done by hand. As with many constructions based upon gadgets, checking correctness of the gadgets is straightforward compared with finding them in the first place, especially when it is not a priori obvious that they exist; this also necessitated significant computer assistance.


\subsection{Organisation}

We begin by proving some preliminary lemmas relating to the normal form of Hermitian matrices and perturbation theory; these will be used throughout the rest of the paper. Then, in Section~\ref{sec:shamlf}, we prove our results about the \shamlf\ problem. First we give the outline of the proof for 2\nobreakdash-local Hamiltonians; the proofs of $\qma$-hardness of various special cases required for this part are deferred to Appendix~\ref{sec:qmalf}. We then proceed to {\sc $\{ZZ\}$\nobreakdash-local Hamiltonian} in Section~\ref{sec:zzlf}, followed by the generalisation to $k$\nobreakdash-local Hamiltonians in Section~\ref{sec:klocallf}. Section~\ref{sec:sham} contains our results on the more general \sham\ problem. Again, we start by giving an outline of the proof, then go on to prove the lemmas required to deal with various special cases in Section~\ref{sec:special}, starting with the Heisenberg model. The proof proceeds to other cases of increasing generality in Sections~\ref{sec:general} and~\ref{sec:local}. We finish the proof in Section~\ref{sec:diagonal} with the case of diagonal Hamiltonians, and conclude with some open questions in Section~\ref{sec:outlook}.


\section{Preliminaries}


\subsection{Notation}

We use $X := \sm{0&1\\1&0}$, $Y := \sm{0&-i\\i&0}$, $Z := \sm{1&0\\0&-1}$ to denote the Pauli matrices, and also define $\sigma^0 := I$, $\sigma^1 := X$, $\sigma^2:=Y$, $\sigma^3:=Z$. Since the Pauli matrices together with the identity matrix form a complete basis for $2\times 2$ matrices, any $k$-qubit matrix $M$ can be decomposed as a weighted sum of tensor products of Pauli matrices; we call the coefficients occurring in this decomposition the Pauli coefficients of $M$. For each $\ell$, $0 \le \ell \le k$, we call the part of $M$ corresponding to Pauli matrices which act non-trivially on exactly $\ell$ qubits the $\ell$\nobreakdash-local part of $M$. In a tensor product of Pauli matrices, we call the number of non-identity matrices the Pauli weight of the overall matrix. We often consider traceless matrices $M$, which have zero 0\nobreakdash-local part (i.e.\ they contain no identity-only term in their Pauli decomposition).

For any $k$-qubit matrix $M$, we let $M_{i_1\dots i_k}$ denote the matrix formed by applying $M$ on qubits $i_1,\dots,i_k$, tensored with the identity elsewhere. For conciseness, we usually follow the condensed-matter convention of writing $AB$ for the two-qubit matrix $A \otimes B$ (so, for example, $XX+YY+ZZ = X \otimes X + Y \otimes Y + Z \otimes Z$). We usually let $n$ denote the number of qubits in the overall Hamiltonian, and $[n]$ denote the set $\{1,\dots,n\}$. $B(\C^d)$ denotes the unit sphere in $\C^d$.

We often use diagrammatic notation of the following form to present our gadgets:
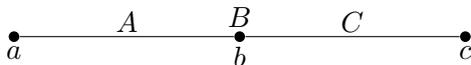
\begin{center}
\begin{tikzpicture}[inner sep=0.5mm,xscale=3]
\node (a) at (0,0) [circle,fill=black,label=south:$a$] {};
\node (b) at (1,0) [circle,fill=black,label=south:$b$,label=north:$B$] {};
\node (c) at (2,0) [circle,fill=black,label=south:$c$] {};
\draw (a) to node[midway,above] {$A$} (b); \draw (b) to node[midway,above] {$C$} (c);
\end{tikzpicture}
\end{center}
Each vertex corresponds to a qubit, with labels given below the vertices. Operators above vertices and edges correspond to interactions; 1\nobreakdash-local interactions sit above vertices and 2\nobreakdash-local interactions sit above edges. So the above diagram illustrates a system of three qubits with an overall Hamiltonian $A_{ab} + B_b + C_{bc}$.


\subsection{Normal form for Hamiltonians}
\label{sec:normal}

Our task will be greatly simplified by the ability to transform any two-qubit Hermitian matrix $H$ into a standard normal form using conjugation by single-qubit unitaries. The normal form we use is essentially the same as one described by Horodecki and Horodecki~\cite{horodecki96}, D\"ur et al.~\cite{dur01} and Bennett et al.~\cite{bennett02}, except that we insist that the unitaries applied are the same on each qubit. Any traceless two-qubit Hermitian matrix $H$ can be written as
\[ H = \sum_{i,j=1}^3 M_{ij} \sigma^i \otimes \sigma^j + \sum_{k=1}^3 v_k \sigma^k \otimes I + w_k I \otimes \sigma^k \]
for some coefficients $M_{ij}$, $v_k$, $w_k$. Write $M(H)$ for the $3 \times 3$ matrix $M$ occurring in this decomposition. In the entanglement theory literature, $M$ is sometimes known as the correlation matrix. Also define the {\em Pauli rank} of $H$ to be the rank of $M(H)$. We observe that if $H$ is symmetric (resp.\ antisymmetric) under exchange of the two qubits on which it acts, $M(H)$ is a symmetric (resp.\ skew-symmetric) matrix. Further, in the symmetric case $v_k = w_k$; in the antisymmetric case $v_k = -w_k$. Using this notation, we can state two lemmas about conjugation by single-qubit unitaries; we defer the proofs to Appendix~\ref{app:normal}.

\begin{lem}
\label{lem:rotate}
Let $H$ be a traceless 2-qubit Hermitian matrix and write
\[ H = \sum_{i,j=1}^3 M_{ij} \sigma^i \otimes \sigma^j + \sum_{k=1}^3 v_k \sigma^k \otimes I + w_k I \otimes \sigma^k. \]
Then, for any orthogonal matrix $R \in SO(3)$, there exists $U \in SU(2)$ such that
\[ U^{\otimes 2} H (U^\dag)^{\otimes 2} = \sum_{i,j=1}^3 (RMR^T)_{ij} \sigma^i \otimes \sigma^j + \sum_{k=1}^3 (Rv)_k \sigma^k \otimes I + (Rw)_k I \otimes \sigma^k. \]
\end{lem}

Note that as compared with the normal form of~\cite{horodecki96,dur01,bennett02}, here the unitaries acting on each qubit of $H$ are the same. This will be important later because mapping $H \mapsto U^{\otimes 2} H (U^\dag)^{\otimes 2}$ does not change the eigenvalues of any Hamiltonian produced only from terms of the type $H$, as
\[ \sum_{i\neq j} \alpha_{ij} (U^{\otimes 2} H (U^\dag)^{\otimes 2})_{ij} = U^{\otimes n} \left(\sum_{i\neq j} \alpha_{ij} H_{ij}\right) (U^\dag)^{\otimes n}. \]
We use Lemma~\ref{lem:rotate} to give a normal form for two special cases. These will be crucial for the proofs of Theorems \ref{thm:shamlf} and \ref{thm:sham}.

\begin{lem}
\label{lem:normalform}
Let $H$ be a traceless 2-qubit Hermitian matrix. If $H$ is symmetric under exchanging the two qubits on which it acts, there exists $U \in SU(2)$ such that
\[ U^{\otimes 2} H (U^\dag)^{\otimes 2} = \sum_{i=1}^3 \alpha_i \sigma^i \otimes \sigma^i + \sum_{j=1}^3 \beta_j (\sigma^j \otimes I + I \otimes \sigma^j), \]
for some real coefficients $\alpha_i$, $\beta_j$. If $H$ is antisymmetric under this exchange, there exists $U \in SU(2)$ and $i \neq j$ such that
\[ U^{\otimes 2} H (U^\dag)^{\otimes 2} = \alpha (\sigma^i \otimes \sigma^j - \sigma^j \otimes \sigma^i) + \sum_{k=1}^3 \beta_k (\sigma^k \otimes I - I \otimes \sigma^k), \]
for some real coefficients $\alpha$, $\beta_k$.
\end{lem}

\subsection{Perturbation theory}

In order to prove our main results, we will need to use several reductions from one class of Hamiltonian to another. In particular, we will use the concept of {\em perturbative gadgets} introduced by Kempe, Kitaev and Regev~\cite{kempe06}. These gadgets allow one to effectively simulate an interaction to which one does not immediately have access. Intuitively, this is similar to producing an effective constraint of a certain type in a CSP by applying constraints of a different type. The main result we will use is a powerful theorem of Oliveira and Terhal~\cite{oliveira08}, which we will encapsulate below as a pair of corollaries tailored to our needs.

Following the terminology of these authors, we distinguish two ways of using multiple gadgets: in series and in parallel. In the first case, we add a sequence of strong interactions, one after the other, to perform a number of encoding operations. Each such encoding increases the norm of the Hamiltonian by a polynomial factor, so we can only perform a constant number of these operations. In the second case, we apply multiple strong interactions at once to different ancilla qubits to simulate a number of interactions simultaneously. In this case, one parallel use of an arbitrary number of gadgets only increases the norm by a polynomial factor; however, one has to take care that the gadgets do not interfere with each other. In our case, as with previous work, this will not be a problem as long as we ensure that all strong interactions occurring in parallel take place across different subsets of qubits. This point is discussed further in Appendix \ref{sec:qmalf}; for detailed justifications, see~\cite{oliveira08,bravyi08,cao14,piddock15}.

Fix Hamiltonians $H$ and $V$ on some Hilbert space $\mathcal{L}$ and set $\widetilde{H} = H + V$. We think of $V$ as a small perturbation of $H$ (i.e.\ $\|V\| \ll \|H\|$). For any matrix $M$, let $M_{<\delta}$ be the restriction of $M$ to the subspace spanned by eigenvectors of $M$ of eigenvalue less than $\delta$. Let $\mathcal{L}_{-}$ (resp.\ $\mathcal{L}_{+}$) be the subspace spanned by eigenvectors of $H$ whose corresponding eigenvalues are less than (resp.\ greater than) $\lambda_\ast$, for some cutoff $\lambda_\ast$. Further define $\Pi_-$, $\Pi_+$ to be the projectors onto $\mathcal{L}_{-}$, $\mathcal{L}_{+}$, and for any matrix $M$ on $\mathcal{L}$, set
\[ M_- = \Pi_- M \Pi_-,\; M_{-+} = \Pi_- M \Pi_+,\; M_{+-} = \Pi_+ M \Pi_-,\; M_+ = \Pi_+ M \Pi_+. \]
%
Define the ``self-energy'' operator $\Sigma_-(z)$, $z \in \C$, as
\[ \Sigma_-(z) := H_- + V_- + V_{-+} G_+(I_+ - V_+ G_+)^{-1} V_{+-}, \]
where $G$ is the resolvent, defined by $G := (zI-H)^{-1}$ (whose dependence on $z$ is implicit). $\Sigma_-(z)$ has the series expansion
\be \label{eq:series} \Sigma_-(z) = H_- + V_- + V_{-+} G_+ V_{+-} + V_{-+} G_+ V_+ G_+ V_{+-} + V_{-+} G_+ V_+ G_+ V_+ G_+ V_{+-} + \dots \ee
If $H$ is proportional to a projector (i.e.\ $H = \Delta \Pi_+$ for some $\Delta \in \R$, where we assume $0 < \lambda_\ast < \Delta$), then $H_-=0$, $G_+ = (z-\Delta)^{-1} I_{\mathcal{L}_+}$, and we have
\be \label{eq:series2} \Sigma_-(z) = V_- +  \frac{ V_{-+} V_{+-}} {z-\Delta} + \frac{V_{-+} V_+ V_{+-}}{(z-\Delta)^2} + \frac{V_{-+} V_+^2 V_{+-}}{(z-\Delta)^3} + \dots \ee

\begin{thm}[Oliveira and Terhal~\cite{oliveira08}, Theorem A.1]
\label{thm:opnorm}
Let $H$ be a Hamiltonian such that no eigenvalues of $H$ lie between $\lambda_{-} := \lambda_{\ast} - \Delta/2$ and $\lambda_+ := \lambda_\ast + \Delta/2$, for some $\Delta$. Let $\widetilde{H} = H + V$ where $\|V\| \le \Delta/2$.

Assume there exists a Hamiltonian $\heff$ such that $\spec(\heff) \subseteq [a,b]$ for some $a<b$, and $\heff = \Pi_{-} \heff \Pi_{-}$. 
Let $D_r$ be a disc of radius $r$ in the complex plane centred at $c = (b+a)/2$, and let $r$ and $\epsilon$ be such that $b + \epsilon < c + r < \lambda_\ast$, and for all $z \in D_r$ we have $\|\Sigma_-(z) - \heff\| \le \epsilon$. Let $w = (b-a)/2$. Then
\[ \|\widetilde{H}_{<\lambda_\ast} - \heff\| \le \frac{3(\|\heff\|+\epsilon)\|V\|}{\lambda_+ - \|\heff\| - \epsilon} + \frac{r(r+c)\epsilon}{(r-w)(r-w-\epsilon)}. \]
\end{thm}

Theorem~\ref{thm:opnorm} has the following easy corollary.

\begin{cor}
\label{cor:perturb}
Let $H$ be a Hamiltonian such that no eigenvalues of $H$ lie between $\lambda_{-} := \lambda_{\ast} - \Delta/2$ and $\lambda_+ := \lambda_\ast + \Delta/2$, for some $\Delta<\lambda_\ast$. Let $\widetilde{H} = H + V$ where $\|V\| \le \Delta/2$. Let $\heff$ be a Hamiltonian obtained by truncating the series expansion (\ref{eq:series}) of $\Sigma_-(0)$ after some finite number of terms. Fix $\epsilon < \|\heff\|/2$ and assume that:
\begin{itemize}
\item $\|\heff - \Sigma_-(z)\| \le \epsilon$, for all $z \in \C$ such that $|z| \le 2\|\heff\|$;
\item $\|\heff\| < \lambda_+/3$.
\end{itemize}
Then
\[ \|\widetilde{H}_{<\lambda_\ast} - \heff\| \le \frac{9\|V\|\|\heff\|}{\lambda_+} + 8 \epsilon. \]
\end{cor}

\begin{proof}
First observe that each term in (\ref{eq:series}) acts only on $\mathcal{L}_-$, so $\heff = \Pi_{-} \heff \Pi_{-}$. Take $b = \|\heff\|$, $a = -\|\heff\|$, $r = 2\|\heff\|$ in Theorem~\ref{thm:opnorm}. Then $c=0$, $w=\|\heff\|$. Further, $b + \epsilon < c + r < \lambda_\ast$ as required. Therefore,
\[ \|\widetilde{H}_{<\lambda_\ast} - \heff\| \le \frac{3(\|\heff\|+\epsilon)\|V\|}{\lambda_+ - (\|\heff\| + \epsilon)} + \frac{4\|\heff\|^2\epsilon}{\|\heff\|(\|\heff\|- \epsilon)} \]
and the claim follows by using $\epsilon < \|\heff\|/2$ and $\|\heff\| < \lambda_+/3$.
\end{proof}

We highlight the following first-order special case of Corollary~\ref{cor:perturb}.

\begin{cor}
\label{cor:zerothorder}
Let $H$ be a Hamiltonian such that $\lambda_{\min}(H)=0$ and the next smallest non-zero eigenvalue of $H$ is 1, and let $V$ be an arbitrary Hamiltonian such that $\|V\| \ge 1$. Write $V_- = \Pi_- V \Pi_-$, where $\Pi_-$ is the projector onto the nullspace of $H$. Further take $\Delta = \delta \|V\|^2$ for some $\delta \ge 4$, and let $\widetilde{H} = \Delta H + V$.
Then
\[ \|\widetilde{H}_{<\Delta/2} - V_-\| \le 41/\delta. \]
\end{cor}

\begin{proof}
We take $\heff = V_-$ (since $H_- = 0$) in Corollary~\ref{cor:perturb}, for which
\begin{align*}
\|\Sigma_-(z) - \heff\| &= \| V_{-+} G_+ V_{+-} + V_{-+} G_+ V_+ G_+ V_{+-} + V_{-+} G_+ V_+ G_+ V_+ G_+ V_{+-} + \dots \|\\
&\le \sum_{i=1}^{\infty} \|V\|^{i+1} \|G_+\|^i.
\end{align*}
Let $z \in \C$ satisfy $|z| \le 2\|V\|$ and let $\lambda_i(H_+)$ denote the $i$'th eigenvalue of $H_+$. Then
\[ \|G_+\| = (\min_i |z-\Delta \lambda_i(H_+)|)^{-1} \le \|V\|^{-1}(\delta\|V\|-2)^{-1}, \]
so we get that for all such $z$
\[ \|\Sigma_-(z) - \heff\| \le \sum_{i=1}^{\infty} \|V\|^{i+1} \|G_+\|^i \le \|V\| \sum_{i=1}^\infty (\delta\|V\|-2)^{-i} = \frac{\|V\|}{\delta \|V\|-3}. \]
Thus, taking $\epsilon = \|V\|/(\delta \|V\|-3)$ in Corollary~\ref{cor:perturb} and noting that $\|\heff\| = \|V_-\| \le \|V\|$, we obtain
\[ \|\widetilde{H}_{<\Delta/2} - V_-\| \le \frac{9\|V_-\|}{\delta \|V\|} + \frac{8\|V\|}{\delta \|V\|-3} \le \frac{41}{\delta}. \]
The last inequality follows from $\|V\|/(\delta\|V\|-3) \le 4/\delta$, which is equivalent to $\delta\|V\| \ge 4$ and holds because we assumed that $\delta \ge 4$, $\|V\| \ge 1$.
\end{proof}

Corollaries~\ref{cor:perturb} and~\ref{cor:zerothorder}, both of which are direct consequences of Theorem~\ref{thm:opnorm} of Oliveira and Terhal, will be the underlying technical tools which we use in the rest of the paper. Corollary~\ref{cor:perturb} will be used for our results on \shamlf, where we use second-order perturbation theory, while Corollary~\ref{cor:zerothorder} will be used for \sham, which needs only first-order perturbation theory. For readability, we often apply these corollaries without specifically verifying their preconditions. This is justified because of the freedom we have in reweighting terms in the Hamiltonian by arbitrary $\poly(n)$ scaling factors, and adding identity terms. For example, when applying Corollary \ref{cor:zerothorder} to produce an effective Hamiltonian $V_-$ from Hamiltonians $H$ and $V$, we can start by scaling $V$ such that $1 \le \|V\| \le \poly(n)$, and take large enough $\Delta \le \poly(n)$ to infer that $\|\widetilde{H}_{<\Delta/2} - V_-\| \le 1/\poly(n)$.

Corollary~\ref{cor:zerothorder} is similar to the ``Projection Lemma'' of Kempe, Kitaev and Regev~\cite{kempe06}, but improves it by showing that the low-energy subspace of $\widetilde{H}$ is actually close to $V_-$ in operator norm, rather than the two matrices just having a similar lowest eigenvalue. This will be important for us as we will need to encode data in this subspace. It is immediate that Corollary~\ref{cor:zerothorder} can be applied a constant number of times in series (which does not seem obvious from the result of~\cite{kempe06}).


\section{\texorpdfstring{\shamlf}{S-Hamiltonian with local terms}}
\label{sec:shamlf}

We begin by studying the \shamlf\ problem. Recall that in this problem, we assume that in addition to $\mathcal{S}$ we have access to interactions $\{X,Y,Z\}$, or in other words arbitrary single-qubit matrices. We make the following observations about this problem.

\begin{itemize}
\item We can perform an arbitrary global relabelling of the Pauli matrices occurring in $\mathcal{S}$ without changing the complexity of \shamlf.
\item We can assume that each matrix in $\mathcal{S}$ does not contain any 1\nobreakdash-local terms. Any such terms can be removed by adding or subtracting arbitrary 1\nobreakdash-local interactions. Thus, in the case where every matrix in $\mathcal{S}$ acts on at most 2 qubits, we can assume that every matrix acts on exactly 2 qubits.
\item We can assume that every 2-qubit matrix $H \in \mathcal{S}$ is either symmetric or antisymmetric under interchange of the two qubits on which it acts. That is, for all $H \in \mathcal{S}$, $M(H)$ is either symmetric or skew-symmetric. This holds because, given access to $H$, we can implement the two matrices $H^+ = (H + FHF)/2$ and $H^- = (H-FHF)/2$, where $F$ is the swap operator, simply by applying $H$ in both the normal direction and in reverse. $H^+$ is symmetric, and $H^-$ is antisymmetric. Further, we have lost nothing by doing this, as $H^+ + H^- = H$.
\end{itemize}

The following lemmas, which characterise the complexity of special cases of this problem, will in fact allow us to characterise {\em all} cases of the problem. We prove the lemmas in Appendix~\ref{sec:qmalf}.

\begin{lem}
\label{lem:xxzzhamlf}
For any fixed $\gamma \neq 0$, {\sc $\{XX + \gamma ZZ\}$-Hamiltonian with local terms} is $\qma$-complete. This holds even if all 2-qubit interactions have the same weight and are restricted to the edges of a 2d square lattice.
\end{lem}

\begin{lem}
\label{lem:xxyyzzhamlf}
For any fixed $\beta,\gamma \neq 0$, {\sc $\{XX + \beta YY + \gamma ZZ\}$-Hamiltonian with local terms} is $\qma$-complete. This holds even if all 2-qubit interactions have the same weight and are restricted to the edges of a 2d square lattice.
\end{lem}


\begin{lem}
\label{lem:xzskewlf}
{\sc $\{XZ - ZX\}$-Hamiltonian with local terms} is $\qma$-complete. This holds even if all 2-qubit interactions have the same weight and are restricted to the edges of a 2d square lattice.
\end{lem}

Special cases of Lemmas~\ref{lem:xxzzhamlf} and~\ref{lem:xxyyzzhamlf} were previously proven by Schuch and Verstraete~\cite{schuch09}. The following lemma, which was also essentially proven in~\cite{schuch09} (see also~\cite{biamonte08}), will be useful as well.

\begin{lem}
\label{lem:xxandzzlf}
Let $\alpha$ and $\beta$ be arbitrary fixed non-zero real numbers. {\sc $\{XX,ZZ\}$-Hamiltonian with local terms} is $\qma$-complete, even if all 2-qubit interactions are restricted to the edges of a 2d square lattice, the $XX$ terms all have weight $\alpha$, and the $ZZ$ terms all have weight $\beta$.
\end{lem}


The final lemma we will need characterises an easier special case:

\begin{lem}
\label{lem:zzhamlf}
{\sc $\{ZZ\}$-Hamiltonian with local terms} is $\stoqma$-complete.
\end{lem}

Assuming Lemmas~\ref{lem:xxzzhamlf}--\ref{lem:zzhamlf}, we now prove a crucial proposition which is the 2\nobreakdash-local special case of Theorem~\ref{thm:shamlf}.

\begin{prop}
\label{prop:2shamlf}
Let $\mathcal{S}$ be an arbitrary fixed subset of 2-qubit Hermitian matrices, and let $\mathcal{S}'$ be the subset formed by removing the 1\nobreakdash-local part from each element of $\mathcal{S}$, and then deleting all 0\nobreakdash-local matrices from the resulting set. Then:
\begin{itemize}
\item If $\mathcal{S}'$ is empty, \shamlf\ is in $\ptime$;
\item Otherwise, if there exists $U \in SU(2)$ such that $U^{\otimes 2} H_i (U^{\dag})^{\otimes 2} = \alpha_i Z^{\otimes 2}$ for all $H_i \in \mathcal{S}'$, where $\alpha_i \in \R$, \shamlf\ is $\stoqma$-complete;
\item Otherwise, \shamlf\ is $\qma$-complete. This holds even if we insist that all 2-qubit interactions in the final Hamiltonian have equal weight and are restricted to the edges of a 2d square lattice.
\end{itemize}
\end{prop}

\begin{proof}
We consider each case in turn. The first case is easy: any Hamiltonian formed only from 1\nobreakdash-local matrices $H_i$ is of the form $H = \sum_i H_i$, and the lowest eigenvalue of $H$ is just the sum of the lowest eigenvalues of the individual matrices $H_i$, which can be calculated efficiently.

For the second case, if there exists such a $U$, then it can be found in time that does not depend on the problem instance. By conjugating all $H_i \in \mathcal{S}'$ by $U^{\otimes 2}$ and rescaling, \shamlf\ is equivalent to {\sc $\{ZZ\}$-Hamiltonian with local terms}. The claim thus follows from Lemma~\ref{lem:zzhamlf} and the discussion in Section~\ref{sec:results}.

For the third case, we first observe that, as a special case of the {\sc Local Hamiltonian} problem, \shamlf\ is clearly in $\qma$. To prove $\qma$-hardness, we split into subcases. First consider the subcase where there exists $H \in \mathcal{S}'$ such that $H$ has Pauli rank at least 2. Then, if $M(H)$ is symmetric, up to rescaling and potentially relabelling Pauli matrices, by Lemma~\ref{lem:normalform} the normal form of $H$ is either of the form $XX + \gamma ZZ$ or $XX + \beta YY + \gamma ZZ$. So by Lemmas~\ref{lem:xxzzhamlf} and~\ref{lem:xxyyzzhamlf}, \shamlf\ is $\qma$-complete. If $M(H)$ is skew-symmetric, then after bringing $H$ to normal form, $H = XZ - ZX$ (up to rescaling and relabelling Pauli matrices), so \shamlf\ is $\qma$-complete by Lemma~\ref{lem:xzskewlf}. Every element in $\mathcal{S}$ that has Pauli rank 0 has been removed, as these are precisely those matrices which are 1\nobreakdash-local. So the only subcase we have left to consider is where each element in $\mathcal{S}'$ has Pauli rank 1, but there does not exist $U$ such that $U^{\otimes 2} H_i (U^{\dag})^{\otimes 2} = \alpha_i Z^{\otimes 2}$ for all $H_i \in \mathcal{S}'$ (as otherwise we would be in the second case). In this subcase there must exist a pair $i\neq j$ and a unitary $U$ such that $U^{\otimes 2} H_i (U^\dag)^{\otimes 2}$ is diagonal, but $U^{\otimes 2} H_j (U^\dag)^{\otimes 2}$ is not diagonal. As $H_i$ and $H_j$ have Pauli rank 1, $M(H_i)$ and $M(H_j)$ are symmetric. By applying $U$ and rescaling, we can assume that
\[ H_i = ZZ, \;\;\;\; H_j = (\alpha X + \beta Y + \gamma Z)(\alpha X + \beta Y + \gamma Z) \]
for some real $\alpha$, $\beta$, $\gamma$ where at least one of $\alpha$ or $\beta$ is non-zero. So, by rescaling $H_j$, we can assume that $\alpha^2 + \beta^2 = 1$. There exists an SO(3) rotation $R$ which maps $(\alpha,\beta,\gamma)$ to $(1,0,\gamma)$ while leaving $(0,0,1)$ unchanged. Therefore, by Lemma~\ref{lem:rotate}, there exists a unitary $V$ such that $V^{\otimes 2} H_j (V^\dag)^{\otimes 2} = (X + \gamma Z)^{\otimes 2}$ and also $V^{\otimes 2} H_i (V^\dag)^{\otimes 2} = ZZ$. If $\gamma=0$, we have the interactions $XX$ and $ZZ$, so by Lemma~\ref{lem:xxandzzlf} this case is also $\qma$-complete. If $\gamma \neq 0$, then by rescaling and subtracting $H_i$ from $H_j$, we can make the interaction $XX + \gamma(ZX+XZ)$. This is local-unitarily equivalent to $XX + \gamma' ZZ$ for some $\gamma' \neq 0$. So this case is $\qma$-complete by Lemma~\ref{lem:xxzzhamlf}, completing the proof.
\end{proof}

The $\qma$-hardness lemmas above are proven using perturbation theory, via similar techniques to previous work of Biamonte and Love~\cite{biamonte08}, Oliveira and Terhal~\cite{oliveira08}, and Schuch and Verstraete~\cite{schuch09}. We defer the proofs to Appendix~\ref{sec:qmalf}. Lemma~\ref{lem:zzhamlf}, by contrast, is proven by showing directly that {\sc $\{ZZ\}$-Hamiltonian with local terms} reduces to {\sc $\{ZZ,X\}$-Hamiltonian}.


\subsection{The case of \texorpdfstring{{\sc $\{ZZ\}$-Hamiltonian with local terms}}{ZZ-Hamiltonian with local terms}}
\label{sec:zzlf}

We now consider the case of 2\nobreakdash-local Hamiltonians whose 2\nobreakdash-local terms are all of the form $ZZ$.

\begin{replem}{lem:zzhamlf}
{\sc $\{ZZ\}$-Hamiltonian with local terms} is $\stoqma$-complete.
\end{replem}

\begin{proof}
The general Ising model with transverse fields, {\sc $\{ZZ,X\}$-Hamiltonian}, is a special case of {\sc $\{ZZ\}$-Hamiltonian with local terms} (cf.\ (\ref{eq:tim})), so by Theorem \ref{thm:bh} {\sc $\{ZZ\}$-Hamiltonian with local terms} is $\stoqma$-hard. We now show that {\sc $\{ZZ\}$-Hamiltonian with local terms} in fact reduces to {\sc $\{ZZ,X\}$-Hamiltonian}, implying that it is in $\stoqma$. Let $H$ be a Hamiltonian on $n$ qubits which is a weighted sum of $ZZ$ terms and arbitrary 1\nobreakdash-local terms. By collecting these 1\nobreakdash-local terms, we can write
\[ H = \sum_{i<j} \alpha_{ij} Z_i Z_j + \sum_k M_k, \]
where the $\alpha_{ij}$ coefficients are arbitrary and each $M_k$ is an arbitrary 1\nobreakdash-local Hermitian matrix acting non-trivially on the $k$'th qubit. By potentially subtracting an $I$ term and writing local $Z$ terms separately, we can rewrite
\[ H = \sum_{i<j} \alpha_{ij} Z_i Z_j + \sum_k \beta_k Z_k + \sum_\ell \gamma_\ell X_\ell + \delta_\ell Y_\ell. \]
For each $\ell$, there exists a unitary $U_\ell$ such that $U_\ell Z U_\ell^\dag = Z$, and
\[ U_\ell (\gamma_\ell X + \delta_\ell Y) U_\ell^\dag = -\sqrt{\gamma^2 + \delta^2} X. \]
By conjugating $H$ by $U_1\otimes \dots \otimes U_n$, we can therefore assume that it is of the form
\[ H = \sum_{i<j} \alpha_{ij} Z_i Z_j + \sum_k \beta_k Z_k + \sum_\ell \gamma_\ell X_\ell. \]
We now show that the 1\nobreakdash-local $Z_k$ interactions can be simulated using only 2\nobreakdash-local interactions. Add an ancilla qubit $a$ and consider the Hamiltonian
\[ H' = \sum_{i<j} \alpha_{ij} Z_i Z_j + \sum_k \beta_k Z_a Z_k + \sum_\ell \gamma_\ell X_\ell. \]
Any normalised eigenvector $\ket{\psi}$ of $H'$ can be decomposed in terms of $\ket{0}$ and $\ket{1}$ on the ancilla qubit, and some other normalised states $\ket{\psi_0}$, $\ket{\psi_1}$ on the rest, as $\ket{\psi} = \eta \ket{0}\ket{\psi_0} + \zeta \ket{1}\ket{\psi_1}$. Thus we have
\begin{eqnarray*}
\bracket{\psi}{H'}{\psi} &=& \left(\eta^*\bra{0}\bra{\psi_0} +\zeta^* \bra{1}\bra{\psi_1} \right)\left(\sum_{i<j} \alpha_{ij} Z_i Z_j + \sum_k \beta_k Z_a Z_k + \sum_\ell \gamma_\ell X_\ell \right) \left(\eta \ket{0}\ket{\psi_0} + \zeta \ket{1}\ket{\psi_1} \right)\\
&=& |\eta|^2 \bra{\psi_0}\left( \sum_{i<j} \alpha_{ij} Z_i Z_j + \sum_k \beta_k Z_k + \sum_\ell \gamma_\ell X_\ell \right) \ket{\psi_0}\\
\\ &+& |\zeta|^2 \bra{\psi_1}\left( \sum_{i<j} \alpha_{ij} Z_i Z_j - \sum_k \beta_k Z_k + \sum_\ell \gamma_\ell X_\ell \right) \ket{\psi_1}.
\end{eqnarray*}
By convexity, the minimum of this quantity over $\ket{\psi}$ is achieved by minimising each of the terms over $\ket{\psi_0}$ and $\ket{\psi_1}$ separately, then taking the minimum of the two. The minimum of the first term is just the lowest eigenvalue of $H$. Let $H^{(1)}$ be the Hamiltonian appearing in the second term, and observe that if we conjugate each qubit of $H^{(1)}$ by $X$, this leaves all terms invariant except the 1\nobreakdash-local $Z$ terms, which have their signs flipped. Therefore, $H^{(1)}$ has the same eigenvalues as $H$, so the lowest eigenvalue of $H'$ is the same as the lowest eigenvalue of $H$. This completes the proof.
\end{proof}


\subsection{Generalisation to \texorpdfstring{$k$}{k}\nobreakdash-local Hamiltonians}
\label{sec:klocallf}

We now complete the proof of Theorem~\ref{thm:shamlf} by generalising Proposition~\ref{prop:2shamlf} to arbitrary $k$\nobreakdash-local Hamiltonians. The proof is based on the following simple corollary of Corollary~\ref{cor:zerothorder}, which allows us to restrict $k$-qubit Hamiltonians to $(k-1)$-qubit Hamiltonians. We use local terms to ``fix'' one qubit in a $k$-qubit interaction, effectively generating a $(k-1)$-qubit interaction on the rest.

\begin{lem}
\label{lem:restrict}
Let $H$ be a Hamiltonian on $n$ qubits, let $a$ be the label of the first qubit, and write $V = I \otimes A + X \otimes B + Y \otimes C + Z \otimes D$ for some Hermitian matrices $A$, $B$, $C$, $D$ on $n-1$ qubits. Consider the Hamiltonian
%
%
\[ \widetilde{H} = V + \Delta \proj{\psi}_a, \]
for arbitrary $\ket{\psi} \in B(\C^2)$, where $\Delta = \delta \|V\|^2$, for arbitrary $\delta\ge4$. This is illustrated by the following diagram, where we generalise our usual notation to allow $k$\nobreakdash-local interactions for $k>2$:
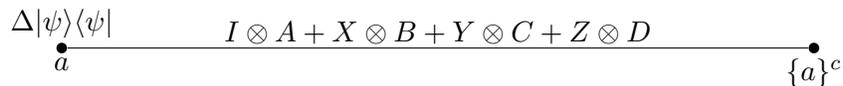
\begin{center}
\begin{tikzpicture}[inner sep=0.5mm,xscale=10]
\node (a) at (0,0) [circle,fill=black,label=south:$a$,label=north:$\Delta \proj{\psi}$] {};
\node (b) at (1,0) [circle,fill=black,label=south:$\{a\}^c$] {};
\draw (a) to node[midway,above] {$I \otimes A + X \otimes B + Y \otimes C + Z \otimes D$} (b);
\end{tikzpicture}
\end{center}
%
Then
\[ \left\|\widetilde{H}_{<\Delta/2} - \proj{\psi^\perp}_a\left(A + \bracket{\psi^\perp}{X}{\psi^\perp} B + \bracket{\psi^\perp}{Y}{\psi^\perp} C + \bracket{\psi^\perp}{Z}{\psi^\perp} D\right) \right\| = O(\delta^{-1}). \]
\end{lem}

\begin{proof}
We use Corollary~\ref{cor:zerothorder}, in whose terminology we have $H = \proj{\psi}_a$ and $V = I \otimes A + X \otimes B + Y \otimes C + Z \otimes D$.
\end{proof}

We observe that, given a Hamiltonian $V$ on $k$ qubits, Lemma~\ref{lem:restrict} allows us to extract arbitrary submatrices given by the Pauli expansion of $V$. Indeed, if we expand $V = I \otimes A + X \otimes B + Y \otimes C + Z \otimes D$, by letting $\ket{\psi}$ be the eigenvector of $X$, $Y$ or $Z$ with eigenvalue $\pm 1$, we can produce the effective interactions $A \pm B$, $A \pm C$ and $A \pm D$ (up to an additive error $O(\delta^{-1})$). By adding/subtracting these matrices we can make each of $\{A,B,C,D\}$. We can then apply this inductively to extract submatrices of $\{A,B,C,D\}$. In particular, for each choice of Pauli matrices on $k-2$ qubits, we can extract the submatrix on the remaining 2 qubits.

We are finally ready to complete the proof of Theorem~\ref{thm:shamlf}, which we restate for convenience.

\begin{repthm}{thm:shamlf}
Let $\mathcal{S}$ be an arbitrary fixed subset of Hermitian matrices on at most $k$ qubits, where $k=O(1)$, and let $\mathcal{S}'$ be the subset formed by removing the 1\nobreakdash-local part from each element of $\mathcal{S}$, and then deleting all 0\nobreakdash-local matrices from the resulting set. Then:
\begin{itemize}
\item If $\mathcal{S}'$ is empty, \shamlf\ is in $\ptime$;
\item Otherwise, if there exists $U \in SU(2)$ such that $U$ locally diagonalises $\mathcal{S}'$, \shamlf\ is $\stoqma$-complete;
\item Otherwise, \shamlf\ is $\qma$-complete. If every matrix in $\mathcal{S}'$ acts on 2 qubits, this holds even if we insist that the 2-qubit interactions in the final Hamiltonian are restricted to the edges of a 2d square lattice and all have equal weight.
\end{itemize}
\end{repthm}

\begin{proof}
The first case is the same argument as in Proposition~\ref{prop:2shamlf}. In the second case, observe that after $U$ is applied to each qubit, each matrix $H_i\in \mathcal{S}$ must be of the form $\sum_{S \subseteq [k]} \alpha^{(i)}_S Z_S$, i.e.\ a weighted sum of tensor products of $Z$ matrices acting on subsets $S$. Also, there must be at least one pair $i$, $S$ with $|S| \ge 2$ such that $\alpha^{(i)}_S \neq 0$ (or all $H_i$ would be 1\nobreakdash-local). We can use Lemma~\ref{lem:restrict} to extract this term, giving us access to an interaction of the form $Z^{\otimes \ell}$ for some $\ell$ between 2 and $k$. If $\ell=2$, we can produce an arbitrary transverse Ising model Hamiltonian, so obtain $\stoqma$-hardness. If $\ell > 2$, we can use Lemma~\ref{lem:restrict} once more to obtain a $Z^{\otimes 2}$ interaction from $Z^{\otimes \ell}$, implying $\stoqma$-hardness again.

We still need to prove that this case of \shamlf\ is contained within $\stoqma$. In order to reduce each $\ell$-body interaction in $\mathcal{S}$ to 2-body interactions, for arbitary constant $\ell > 2$, we rely on work of Biamonte~\cite[Theorem IV.1]{biamonte08a}, who has shown precisely the result we need: namely that the low-energy eigenvalues of any $\ell$-body diagonal interaction can be reproduced exactly by a 2-body diagonal interaction using some additional qubits. Also, using a similar argument to the proof of Lemma~\ref{lem:zzhamlf}, we can assume that the 1\nobreakdash-local part of any Hamiltonian produced using terms in $\mathcal{S}$ and additional arbitrary 1\nobreakdash-local terms is a weighted sum of $Z$ and $X$ terms. Thus the diagonal part of the Hamiltonian can be simulated using 2-body interactions via the techniques of~\cite{biamonte08a}, and the $X$ terms added afterwards. (We remark that subsequent work of de las Cuevas and Cubitt~\cite{delascuevas14} gives a characterisation of all diagonal Hamiltonians for which this idea works.)


%
%
%
%

We finally move on to the third case. Consider an arbitrary set $\mathcal{S}$ of $k$\nobreakdash-local matrices that does not satisfy either of the other two cases and assume that \shamlf\ is {\em not} $\qma$-complete. By Proposition~\ref{prop:2shamlf} and Lemma~\ref{lem:restrict}, this implies that there exists a $U$ such that the 2\nobreakdash-local part of every 2-qubit matrix $H'$ that can be produced by extracting terms from any matrix in $\mathcal{S}$ is diagonalised by $U^{\otimes 2}$. As local unitaries preserve Pauli weights, this means that the 2\nobreakdash-local component of $H'$ is proportional to $(U^\dag Z U)^{\otimes 2} =: M^{\otimes 2}$.

Let $H$ be any element of $\mathcal{S}$, and assume without loss of generality that $H$ acts on exactly $k$ qubits. Fix $d \ge 2$ and let $H^{(d)}$ denote the $d$\nobreakdash-local part of $H$ (i.e.\ the part which is $d$\nobreakdash-local, but not $d'$\nobreakdash-local for any $d' < d$). 
For each $d$-subset $S$ of $[k]$, let $H_S$ be the matrix produced by summing all the terms in the Pauli expansion of $H$ which are the identity on the complement of $S$. Then $\sum_S H_S^{(d)} = H^{(d)}$. Consider each term $H_S$ that has non-zero $d$\nobreakdash-local part. Now, for each 2-subset $T \subseteq S$, and each $w \in \{0,\dots,3\}^{|S|-2}$, let $A_{(T)}^{(w)}$ be the 2-qubit matrix produced by extracting terms in the Pauli decomposition of $H_S$ which correspond to the string $w$ on the set $S\backslash T$. Then $H_S = \sum_w A_{(T)}^{(w)} \otimes \sigma_w$, where $\sigma_w = \bigotimes_{i=1}^{|S|-2} \sigma^{w_i}$. For all $w$, the 2\nobreakdash-local part of $A_{(T)}^{(w)}$ is equal to $\alpha_{T,w} M^{\otimes 2}$ for some coefficient $\alpha_{T,w}$, otherwise \shamlf\ would be $\qma$-complete. Further, if we sum the 2\nobreakdash-local part of $A_{(T)}^{(w)}$ over strings $w$ containing no zeroes (call this set $F$), we get precisely $H_S^{(d)}$, i.e.\
\[ H^{(d)}_S = \sum_{w\in F} A_{(T)}^{(w)} \otimes \sigma_w = \sum_{w\in F} \alpha_{T,w} M^{\otimes 2} \otimes \sigma_w = M^{\otimes 2} \otimes \left( \sum_{w\in F} \alpha_{T,w} \sigma_w \right) =: M^{\otimes 2} \otimes N_T. \]
This holds for each 2-subset $T$, so as $M^{\otimes 2}$ is product, $H^{(d)}_S$ is product across $|S|$-wise splits into qubits, and in particular is proportional to $M^{\otimes d}$. Summing over $S$, we get that $H^{(d)}$ is a linear combination of $M^{\otimes d}$'s and hence is diagonalised by $U^{\otimes k}$. This argument works for any $d \ge 2$; summing over $d$, we get that $H$ is diagonalised by $U^{\otimes k}$. The same argument works for every other element of $\mathcal{S}$, so we get that there exists a $U$ such that every element in $\mathcal{S}$ is locally diagonalised by $U$, so $\mathcal{S}$ fits into the second case stated in the Theorem. This completes the proof.
\end{proof}


\section{The case without local terms}
\label{sec:sham}

We now move on to the more general case of the \sham\ problem, where $\mathcal{S}$ now does not necessarily include all 1\nobreakdash-local terms. In this setting, we no longer have access to the same perturbative gadgets. The proof of this case once again uses the normal form for two-qubit Hamiltonians, and reductions from a few simpler special cases. This time, however, rather than using second-order perturbation theory, the reductions are based on encoding one logical qubit in multiple physical qubits and then using first-order perturbation theory (Corollary~\ref{cor:zerothorder}). The two most important special cases we need to consider are the Heisenberg and XY models, but we also deal with a skew-symmetric case:

\begin{lem}
\label{lem:heisenberg}
{\sc $\{XX+YY+ZZ\}$-Hamiltonian} is $\qma$-complete.
\end{lem}

\begin{lem}
\label{lem:xy}
{\sc $\{XX+YY\}$-Hamiltonian} is $\qma$-complete.
\end{lem}

\begin{lem}
\label{lem:xzskew}
{\sc $\{XZ-ZX\}$-Hamiltonian} is $\qma$-complete.
\end{lem}

See Sections~\ref{sec:heisenberg}, \ref{sec:xy} and~\ref{sec:skew}, respectively, for the proofs of these lemmas. Based on reductions from these models, we can prove $\qma$-completeness of more general cases.

\begin{lem}
\label{lem:xyz}
For any real $\beta$, $\gamma$ such that at least one of $\beta$ and $\gamma$ is non-zero, {\sc $\{XX+\beta YY+\gamma ZZ\}$-Hamiltonian} is $\qma$-complete.
\end{lem}

\begin{lem}
\label{lem:extractlts}
For any $\beta$, $\gamma$ such that at least one of $\beta$ and $\gamma$ is non-zero, and any single-qubit Hermitian matrix $A$, {\sc $\{XX+\beta YY+\gamma ZZ + AI + IA\}$-Hamiltonian} is $\qma$-complete.
\end{lem}

\begin{lem}
\label{lem:extractltsskew}
For any single-qubit Hermitian matrix $A$, {\sc $\{XZ-ZX + AI - IA\}$-Hamiltonian} is $\qma$-complete.
\end{lem}

These lemmas are proven in Sections~\ref{sec:xxayy}, \ref{sec:xxayybzz} and~\ref{sec:local}, respectively. We will also need some reductions from cases which are unlikely to be $\qma$-complete; these are also proven in Section~\ref{sec:local}.

\begin{lem}
\label{lem:extractlts2}
For any single-qubit Hermitian matrix $A$ such that $A$ does not commute with $Z$, {\sc $\{ZZ, X, Z\}$-Hamiltonian} reduces to {\sc $\{ZZ+ AI + IA\}$-Hamiltonian}.
\end{lem}

\begin{lem}
\label{lem:extractlts3}
For any single-qubit Hermitian matrix $A$ such that $A$ does not commute with $Z$, {\sc $\{ZZ, X, Z\}$-Hamiltonian} reduces to {\sc $\{ZZ, AI-IA\}$-Hamiltonian}.
\end{lem}

Finally, we consider the purely classical case of diagonal matrices.

\begin{lem}
\label{lem:np}
Let $\mathcal{S}$ be a set of diagonal Hermitian matrices on at most 2 qubits. Then, if every matrix in $\mathcal{S}$ is 1\nobreakdash-local, {\sc $\mathcal{S}$-Hamiltonian} is in $\ptime$. Otherwise, {\sc $\mathcal{S}$-Hamiltonian} is $\nptime$-complete.
\end{lem}

This final lemma is a special case of a result of Jonsson~\cite{jonsson00}, who classifies the complexity of all maximisation variants of boolean constraint satisfaction problems where both positive and negative weights are allowed. We include a simple direct proof in Section~\ref{sec:diagonal}. Based on all the above lemmas, we are ready to prove Theorem~\ref{thm:sham}, which we restate as follows.

\begin{repthm}{thm:sham}
Let $\mathcal{S}$ be an arbitrary fixed subset of Hermitian matrices on at most 2 qubits. Then:
\begin{itemize}
\item If every matrix in $\mathcal{S}$ is 1\nobreakdash-local, \sham\ is in $\ptime$;
\item Otherwise, if there exists $U \in SU(2)$ such that $U$ locally diagonalises $\mathcal{S}$, then \sham\ is $\nptime$-complete;
\item Otherwise, if there exists $U \in SU(2)$ such that, for each 2-qubit matrix $H_i \in \mathcal{S}$, $U^{\otimes 2} H_i (U^{\dag})^{\otimes 2} = \alpha_i Z^{\otimes 2} + A_i\otimes I + I\otimes B_i$, where $\alpha_i \in \R$ and $A_i$, $B_i$ are arbitrary single-qubit Hermitian matrices, then \sham\ is $\stoqma$-complete;
\item Otherwise, \sham\ is $\qma$-complete.
\end{itemize}
\end{repthm}

Before we prove the theorem, we recall (from the discussion at the start of Section~\ref{sec:shamlf}) that we can assume that each 2-qubit matrix $H \in \mathcal{S}$ is either symmetric or antisymmetric with respect to swapping the two qubits on which it acts.

\begin{proof}
The first case is clear: any Hamiltonian that can be made from $\mathcal{S}$ is of the form $H = \sum_i H_i$ for 1\nobreakdash-local matrices $H_i$, so the lowest eigenvalue of $H$ is just the sum of the lowest eigenvalues of the individual matrices $H_i$, which can be calculated efficiently.

For the second case, if there exists such a $U$, we can apply it and we are left with a set of diagonal matrices where at least one is not 1\nobreakdash-local (or we would be in the first case). The claim then follows from Lemma~\ref{lem:np}.

For the third case, the problem is clearly no harder than {\sc $\{ZZ\}$-Hamiltonian with local terms}, so is contained within $\stoqma$ by Lemma~\ref{lem:zzhamlf}. To prove $\stoqma$-hardness, first note that after applying $U$, there must exist a matrix $H_i \in \mathcal{S}$ of the form $\alpha_i ZZ + A_i I + IB_i$ with $\alpha_i \neq 0$, or we would be in the first case. Symmetrising and rescaling, we can make a matrix of the form $ZZ + \beta (A I + IA)$ (where $\beta$ or $A$ might be zero). If $A$ does not commute with $Z$, Lemma~\ref{lem:extractlts2} implies that {\sc $\{ZZ, X, Z\}$-Hamiltonian} reduces to {\sc $\mathcal{S}$-Hamiltonian}, so by Theorem \ref{thm:bh} {\sc $\mathcal{S}$-Hamiltonian} is $\stoqma$-hard. Therefore, assume that $A$ commutes with $Z$. As $A$ can be taken to be traceless by adding an overall identity term, this is equivalent to $A$ being proportional to $Z$. As we are not in the second case, there must also either exist a 2-qubit matrix $H_j \in \mathcal{S}$ of the form $\alpha_j ZZ + A_j I + IB_j$, where either $A_j$ or $B_j$ does not commute with $Z$, or a 1-qubit matrix $H_k \in \mathcal{S}$ that does not commute with $Z$. If the latter of these possibilities is true, we can make $IH_k + H_kI$, so it suffices to assume the first possibility is true. Note that it could be the case that $i=j$ or $\alpha_j=0$ (but not both).

First assume that $A_j \neq -B_j$. Then by rescaling and symmetrising both matrices, we can assume we have access to matrices of the form
\[ H_i = ZZ + \alpha(ZI + IZ),\;\;\;\; H_j = \gamma ZZ + BI+IB, \]
where $B$ is a traceless Hermitian matrix that does not commute with $Z$, and $\alpha,\gamma \in \R$. By adding a suitable multiple of $H_i$ to $H_j$ and rescaling, we can produce a matrix $H'$ such that $H' = ZZ + AI + IA$ for some matrix $A$ which does not commute with $Z$. By Lemma~\ref{lem:extractlts2}, this implies that {\sc $\{ZZ, X, Z\}$-Hamiltonian} reduces to {\sc $\mathcal{S}$-Hamiltonian}, so once again {\sc $\mathcal{S}$-Hamiltonian} is $\stoqma$-hard.

On the other hand, if $A_j = -B_j$, we have
\[ H_i = ZZ + \alpha(ZI + IZ),\;\;\;\; H_j = \gamma ZZ + BI-IB, \]
where $B$ is a traceless Hermitian matrix that does not commute with $Z$, and $\alpha$, $\gamma \in \R$. By adding a suitable multiple of $H_i$ to $H_j$, antisymmetrising and rescaling, we can produce a matrix $ZZ + BI-IB$ for some $B$ that does not commute with $Z$. Combining Lemma~\ref{lem:extractlts3} and Theorem \ref{thm:bh} then implies that {\sc $\mathcal{S}$-Hamiltonian} is $\stoqma$-hard.

We finally address the fourth case (the $\qma$-hard case), which is split into two subcases. In the first subcase, assume there exists at least one 2-qubit matrix $H \in \mathcal{S}$ which has Pauli rank at least 2. $M(H)$ can be assumed to be either symmetric or skew-symmetric. If $M(H)$ is symmetric, by Lemma~\ref{lem:normalform} (and possibly relabelling Pauli matrices), using local unitaries $H$ can be mapped to $XX+\beta YY+\gamma ZZ + AI + IA$ for some $\beta$, $\gamma$ such that at least one of them is non-zero, and some single-qubit Hermitian matrix $A$. $\qma$-completeness follows from Lemma~\ref{lem:extractlts}. If $M(H)$ is skew-symmetric, we similarly get $\qma$-completeness from Lemma~\ref{lem:extractltsskew}.

In the second subcase, assume all 2-qubit matrices in $\mathcal{S}$ have Pauli rank 1. There does not exist $U$ such that $U^{\otimes 2} H_i (U^{\dag})^{\otimes 2} = \alpha_i ZZ + A_i I + I B_i$ for all $H_i \in \mathcal{S}$, otherwise we would be in the third case. So in this subcase there must exist a pair $i\neq j$ and a unitary $U$ such that $U^{\otimes 2} H_i^{(2)} (U^\dag)^{\otimes 2}$ is diagonal, but $U^{\otimes 2} H_j^{(2)} (U^\dag)^{\otimes 2}$ is not, where $H_i^{(2)}$ is the 2\nobreakdash-local part of $H_i$. By applying this $U$ and rescaling, we can assume that
\[ H_i = ZZ + \delta(CI + IC), \;\;\;\; H_j = (\alpha X + \beta Y + \gamma Z)^{\otimes 2} + \eta(DI + ID) \]
for some real $\alpha$, $\beta$, $\gamma$, $\delta$, $\eta$ where at least one of $\alpha$ or $\beta$ is non-zero. With the freedom to rescale $H_j$, we can assume that $\alpha^2 + \beta^2 = 1$. There exists an SO(3) rotation $R$ which maps $(\alpha,\beta,\gamma)$ to $(1,0,\gamma)$ while leaving $(0,0,1)$ unchanged. Therefore, by Lemma~\ref{lem:rotate}, there exists a unitary $V$ such that $V^{\otimes 2} H_j (V^\dag)^{\otimes 2} = (X + \gamma Z)^{\otimes 2}$ and also $V^{\otimes 2} H_i^{(2)} (V^\dag)^{\otimes 2} = ZZ$. If $\gamma=0$, we have something of the form
\[ H_i = ZZ + \delta(CI + IC),\;\;\;\; H_j = XX + \eta(DI + ID). \]
Adding these two matrices, relabelling Pauli matrices and using Lemma~\ref{lem:xy} and Lemma~\ref{lem:extractlts}, this case is also $\qma$-complete. If $\gamma \neq 0$, by rescaling and subtracting $H_i$ from $H_j$, we can make a matrix whose 2\nobreakdash-local part is unitarily equivalent to $XX + \gamma' ZZ$ for some $\gamma' \neq 0$, so this case is $\qma$-complete by Lemmas~\ref{lem:xyz} and~\ref{lem:extractlts}.
\end{proof}

It remains to prove all the required lemmas, which we now do in order.


\section{Special cases}
\label{sec:special}

We first prove $\qma$-hardness for the Heisenberg and XY interactions, and the interaction $XZ-ZX$, as needed for Lemmas \ref{lem:heisenberg}--\ref{lem:xzskew}. We will often need to use Corollary~\ref{cor:zerothorder} to effectively project the low-energy part of a Hamiltonian onto a smaller space, up to a small ($1/\poly(n)$) additive error. For readability, we will not include these additive errors in the description that follows. Many of the somewhat technical calculations we perform in this and subsequent sections were carried out using a computer algebra package.

\subsection{The Heisenberg model}
\label{sec:heisenberg}

The first special case we consider is the (general) Heisenberg model. The Heisenberg interaction is the 2-qubit matrix $XX + YY + ZZ$. By adding an irrelevant identity term and rescaling we get the swap gate
\[ F := \frac{1}{2}\left(I + XX + YY + ZZ\right). \]
This is therefore sometimes called the exchange interaction. The Heisenberg model describes Hamiltonians of the form
\[ H = \sum_{i < j} \alpha_{ij} (X_i X_j + Y_i Y_j + Z_i Z_j). \]
If $\alpha_{ij} \ge 0$ for all $i<j$, this is called the antiferromagnetic Heisenberg model; if $\alpha_{ij} \le 0$ for all $(i,j)$, it is called the ferromagnetic Heisenberg model. When considering ground-state properties alone, the ferromagnetic Heisenberg model is trivial as the state $\ket{0}^{\otimes n}$ is a ground state. Schuch and Verstraete proved that determining ground-state energies in the Heisenberg model is $\qma$-hard if one allows arbitrary additional 1\nobreakdash-local terms~\cite{schuch09}. Our task in this section is to prove this claim without any additional 1\nobreakdash-local terms. There has been a significant amount of work done to prove that universal quantum computation can be achieved using the exchange interaction alone (see e.g.~\cite{divincenzo00,kempe01,kempe02,hsieh03} and in particular~\cite{kempe00}). However, this does not seem to immediately imply $\qma$-hardness of the Heisenberg model without additional 1\nobreakdash-local terms.

As $F$ commutes with $U^{\otimes 2}$ for any $U$, $H$ is invariant under conjugation by $U^{\otimes n}$, implying that the projectors onto each of its eigenspaces are also invariant under conjugation by $U^{\otimes n}$. This symmetry means that in order to express an arbitrary Hamiltonian in terms of the Heisenberg model, we will have to encode it somehow. In particular, we would like to encode a qubit in a larger space such that we can generate two non-commuting matrices which encode $X$ and $Z$ on the logical qubit.

The simplest such encoding is to associate a block of three physical qubits with each logical qubit (a similar idea was used in~\cite{kempe00}). To take advantage of the symmetry of the swap operation, we decompose $(\C^2)^{\otimes 3}$ in terms of the symmetric subspace and its complement, i.e.\ the subspaces
\[ S_1 = \spann \left\{\ket{000},\frac{1}{\sqrt{3}}(\ket{001}+\ket{010}+\ket{100}),\;\; \frac{1}{\sqrt{3}}(\ket{110}+\ket{101}+\ket{011}),\;\;\ket{111} \right\}, \]
\begin{multline*} S_2 = \spann \bigg\{ \frac{1}{\sqrt{2}}\left(\ket{01}-\ket{10} \right) \ket{0},  \frac{1}{\sqrt{2}}\left(\ket{01}-\ket{10} \right) \ket{1},\\
 -\sqrt{\frac{2}{3}}\ket{001} + \frac{1}{\sqrt{6}}\left(\ket{01}+\ket{10} \right)\ket{0}, \sqrt{\frac{2}{3}}\ket{110} - \frac{1}{\sqrt{6}}\left(\ket{01}+\ket{10} \right)\ket{1} \bigg\}.
 \end{multline*}
Then it is clear that $F$, applied to any pair of the qubits, leaves $S_1$ invariant. In the case of $S_2$, with respect to the above basis one can explicitly calculate that
 \[ F_{12} = \begin{pmatrix} -1 & 0 & 0 & 0\\ 0 & -1 & 0 & 0\\ 0 & 0 & 1 & 0\\ 0 & 0 & 0 & 1 \end{pmatrix},\;
 F_{13} = \frac{1}{2} \begin{pmatrix} 1 & 0 & \sqrt{3} & 0\\ 0 & 1 & 0 & \sqrt{3}\\ \sqrt{3} & 0 & -1 & 0\\ 0 & \sqrt{3} & 0 & -1 \end{pmatrix},\;
 F_{23} = \frac{1}{2} \begin{pmatrix} 1 & 0 & -\sqrt{3} & 0\\ 0 & 1 & 0 & -\sqrt{3}\\ -\sqrt{3} & 0 & -1 & 0\\ 0 & -\sqrt{3} & 0 & -1 \end{pmatrix}, \]
 where the subscripts here denote the qubits that $F$ acts on. Hence we can implement matrices on $S_2$ of the form
 \[ F_{12} + F_{13} + F_{23} = 0,\;\; -F_{12} = Z\otimes I,\;\; \frac{1}{\sqrt{3}}\left(F_{13} - F_{23} \right) = X\otimes I. \]
 On the whole space $(\C^2)^{\otimes 3}$, the first of these corresponds to the projection onto $S_1$. Using Corollary~\ref{cor:zerothorder}, by applying this interaction with a large but polynomially bounded weight, we can (simultaneously!)\ enforce each of the 3-qubit blocks to be contained within $S_2$. Our $n$ triples of physical qubits thus give us a logical space corresponding to $n$ pairs of logical qubits; within each qubit pair we can apply $Z$ or $X$ to the first qubit. Note that these are not really separate qubits as we cannot address the second qubit.

We now need to implement interactions across pairs of logical qubits. Imagine we have two physical qubit triples, with the first triple labelled 1 to 3, and the second triple labelled 4 to 6. By applying $F$ operators across different pairs of physical qubits, we have 9 potential interactions on the logical space of 4~qubits, split into two blocks of two logical qubits: $(1,2)$ and $(3,4)$ (plus the 6 interactions we already know about, by applying $F$ across pairs in the same triple). By explicit calculation, each choice $(i,j)$ such that $i$ and $j$ are in different triples turns out to give a logical interaction of the form
 \[ M^{(i,j)}_{13} (2F-I)_{24} + I^{\otimes 4}/2. \]
As usual, we can ignore the identity term. We will not write out all of the matrices $M^{(i,j)}$, merely recording that
\[ \frac{3}{2}\left( M^{(1,4)} - M^{(1,5)} - M^{(2,4)} + M^{(2,5)}\right) = XX, \]
\[
  \frac{1}{2}\left( M^{(1,4)} + M^{(1,5)} - 2 M^{(1,6)} + M^{(2,4)} + M^{(2,5)} - 2M^{(2,6)} - 2M^{(3,4)} - 2 M^{(3,5)} + 4M^{(3,6)}\right) = ZZ,
\]
\be \label{II_2F-I} 2\sum_{i=1}^3 \sum_{j=4}^6 M^{(i,j)} = II. \ee
The first two of these mean that we can implement the interactions $XX$ and $ZZ$ across logical qubits $(1,3)$ -- but product with $(2F-I\otimes I)$ across qubits $(2,4)$. In other words, we can implement a logical Hamiltonian of the form
\[ H = \sum_{i=1}^n (\alpha_i X_i + \beta_i Z_i) I_{i'} + \sum_{i< j} (\gamma_{ij} X_i X_j + \delta_{ij} Z_i Z_j) (2F-I)_{i'j'}, \]
where we identify the $i$'th logical qubit pair with indices $(i,i')$. We would like to eliminate the unwanted $(2F-I)$ operators. One way to do this is to force the primed qubits to be in a particular state by very strong $F_{i'j'}$ interactions. Consider adding in the (logical) term
\[ G = \Delta \sum_{i<j} w_{ij} F_{i'j'}  \]
where $w_{ij}$ are some weights and $\Delta$ is very large. We can do this because we can make $I_1 I_3 (2F-I)_{24}$, as in the example (\ref{II_2F-I}) above. If the ground state $\ket{\psi}$ of $G$ is non-degenerate, by Corollary~\ref{cor:zerothorder} the primed qubits will all be effectively projected onto the ground state, and $H$ will become
\[ \widetilde{H} = \sum_{i=1}^n (\alpha_i X_i + \beta_i Z_i) + \sum_{i< j} (\gamma_{ij} X_i X_j + \delta_{ij} Z_i Z_j) \bracket{\psi}{(2F-I)_{i'j'}}{\psi}. \]
We therefore need to find a $G$ such that the ground state is non-degenerate and $\bracket{\psi}{(2F-I)_{i'j'}}{\psi} \neq 0$ for all $i$, $j$ (and also these quantities should be easily computable). In particular, this implies that for all pairs $i' \neq j'$, we need the reduced state $\psi_{\{i',j'\}}$ on that pair of qubits to satisfy $\psi_{\{i',j'\}} \neq I/4$. In order to find such a $G$, we study exactly solvable restricted special cases of the Heisenberg model.


\subsection{Restricted Heisenberg models}
\label{sec:restrict}

For our purposes, we call a model (i.e.\ family of Hamiltonians) \emph{exactly solvable} if the eigenvalues and corresponding eigenvectors of any Hamiltonian in the model can be calculated efficiently. Note that the existence of an efficient description of the ground state $\ket{\psi}$ does not necessarily imply the ability to compute quantities of interest about the model, such as the two-point correlation functions $\bracket{\psi}{\sigma^i_j\sigma^k_\ell}{\psi}$. Only very few restricted versions of the Heisenberg model are known to be exactly solvable. The only ones which we are aware of are:

\begin{itemize}
\item The Heisenberg chain
\[ H = \sum_{i=1}^n X_iX_{i+1} + Y_iY_{i+1} + Z_iZ_{i+1}, \]
with either periodic or non-periodic boundary conditions (i.e.\ a ring or a line). This can be solved using the (highly non-trivial) Bethe ansatz. However, an efficient description of the correlation functions appears not to be known~\cite{korepin07}.

\item A special case of the Majumdar-Ghosh model~\cite{majumdar69}
\[ H = \sum_{i=1}^n X_iX_{i+1} + Y_iY_{i+1} + Z_iZ_{i+1} + \frac{1}{2} \left( X_iX_{i+2} + Y_iY_{i+2} + Z_iZ_{i+2}\right) \]
on a ring. This is known to have a two-dimensional ground space spanned by singlets between pairs of adjacent qubits. There are also generalisations of this model by Klein~\cite{klein82}, whose ground states retain this singlet structure.

\item The Haldane-Shastry model on a ring of $n$ qubits~\cite{haldane88,shastry88}
\[ H = \sum_{i<j} \frac{1}{\sin^2(\pi(j-i)/n)} \left( X_iX_j + Y_iY_j + Z_iZ_j \right). \]
This has an entangled, non-degenerate ground state whose correlation functions are known and given in terms of sine integrals.

\item The Lieb-Mattis model~\cite{lieb62}
\[ H = \sum_{i \in A, j \in B} X_iX_j + Y_iY_j + Z_iZ_j, \]
where $A$ and $B$ are disjoint subsets of qubits. That is, the interaction graph of this model is the complete bipartite graph on $A \times B$. The special case of this model where $|A|=1$ is sometimes known as the Heisenberg star~\cite{richter94} for obvious reasons. This model may seem entirely non-physical as it has infinite-range interactions, but in fact it has been useful in the description of various physical systems (see e.g.~\cite{vanwezel05}). It is known that when $|A|=|B|$, the ground state is non-degenerate and can be written down exactly; indeed, even when $|A| \neq |B|$, everything about the model can be understood.
\end{itemize}

The Heisenberg chain and Majumdar-Ghosh model are both ruled out for our purposes: the former because of the lack of an efficient description of the correlation functions, and the latter because of its two-dimensional ground space. The Haldane-Shastry model could potentially be used, given a sufficiently efficient algorithm for computing the trigonometric integrals involved. However, the Lieb-Mattis model satisfies all our criteria and will be substantially simpler to analyse. For the case $|A|=|B|=n$, which will be sufficient, we have the following lemma, which combines results stated (for example) in~\cite{lieb62,vidal07}. First define
\[ \ket{\psi^n_k} := \frac{1}{\sqrt{\binom{n}{k}}} \sum_{x \in \{0,1\}^n, |x|=k} \ket{x}. \]

\begin{lem}
\label{lem:liebmattis}
Write
\[ H_{LM} = \sum_{i=1}^n \sum_{j=n+1}^{2n} X_i X_j + Y_i Y_j + Z_i Z_j. \]
Then the ground state of $H_{LM}$ is unique and given by
\[ \ket{\phi_{LM}} := \frac{1}{\sqrt{n+1}} \sum_{k=0}^n (-1)^k \ket{\psi^n_k}\ket{\psi^n_{n-k}}. \]
For $i$ and $j$ such that $1 \le i,j \le n$ or $n+1 \le i,j \le 2n$, $\bracket{\phi_{LM}}{F_{ij}}{\phi_{LM}} = 1$. Otherwise, $\bracket{\phi_{LM}}{F_{ij}}{\phi_{LM}} = -2/n$.
\end{lem}

The beautiful proof of Lemma~\ref{lem:liebmattis} is well-known in the condensed-matter theory literature, and the most difficult part (proving uniqueness) was already shown by Lieb and Mattis in their original paper~\cite{lieb62}. However, the ingredients of the proof are somewhat scattered. We therefore present a self-contained proof in Appendix~\ref{app:liebmattis}.

Given Lemma~\ref{lem:liebmattis} and the above discussion, the proof of Lemma~\ref{lem:heisenberg} is essentially immediate. We first (potentially) add one triple of physical qubits to make the total number of logical qubit pairs even and equal to $2n$ for some integer $n$. Then, by Corollary~\ref{cor:zerothorder}, we can effectively implement Hamiltonians of the form
\[ \widetilde{H} = \sum_{k=1}^{2n} \alpha_k X_k + \beta_k Z_k + \sum_{i<j} (\gamma_{ij} X_i X_j + \delta_{ij} Z_i Z_j) \bracket{\phi_{LM}}{(2F-I)_{i'j'}}{\phi_{LM}}. \]
The quantity $\bracket{\phi_{LM}}{(2F-I)_{i'j'}}{\phi_{LM}}$ is non-zero, has magnitude lower bounded by an inverse polynomial in $n$, and is efficiently computable for all pairs $i$, $j$. Indeed, it is either equal to 1 or $-4/n - 1$ depending on $i$ and $j$. So, by rescaling $\gamma_{ij}$ and $\delta_{ij}$ appropriately, we can effectively implement any Hamiltonian of the form
\[ \widetilde{H} = \sum_{k=1}^{2n} \alpha_k X_k + \beta_k Z_k + \sum_{i<j} \gamma_{ij} X_i X_j + \delta_{ij} Z_i Z_j \]
for any choices of $\alpha_k$, $\beta_k$, $\gamma_{ij}$, $\delta_{ij}$. By Lemma~\ref{lem:xxandzzlf}, this suffices for $\qma$-completeness. We have proven the following lemma.

\begin{replem}{lem:heisenberg}
{\sc $\{XX+YY+ZZ\}$-Hamiltonian} is $\qma$-complete.
\end{replem}


\subsubsection{More physically realistic variants of the Heisenberg model}
\label{sec:realistic}

Our construction proving $\qma$-hardness of the general Heisenberg model involves interactions between many pairs of spatially distant qubits, whose weights have differing signs, and also a highly non-planar interaction graph. It is natural to wonder whether one could modify it to be more physically natural, and perhaps only involving interactions on a 2d square lattice, as can be achieved for {\sc $\{XX+YY+ZZ\}$-Hamiltonian with Local Terms}~\cite{schuch09}. In addition, for physical realism one might seek to restrict the interactions within the Hamiltonian to all have weights with the same sign (i.e.\ to consider the ferromagnetic or antiferromagnetic cases of the Heisenberg model). The following observation, which was already made in~\cite{bravyi06} and essentially even in~\cite{lieb62}, shows that $\qma$-hardness is unlikely if we combine both of these constraints.

\begin{obs}
\label{obs:antiferr}
Consider a Hamiltonian $H$ of the form $H = \sum_{i<j} \alpha_{ij} (X_i X_j + Y_i Y_j + Z_i Z_j)$. Then, if $\alpha_{ij} \le 0$ for all $i$, $j$, determining the ground-state energy of $H$ up to inverse-polynomial precision is in $\ptime$. If $\alpha_{ij} \ge 0$ for all $i$, $j$, and the graph of interactions that occur in $H$ is bipartite, determining the ground-state energy of $H$ up to inverse-polynomial precision is in $\stoqma$.
\end{obs}

\begin{proof}
In the first case, a ground state of $H$ is the product state $\ket{0}^{\otimes n}$, so the problem is trivial. In the second case, split the qubits on which $H$ acts into two sets $A$ and $B$ such that all interactions are between $A$ and $B$, and apply $Z$ rotations to the $B$ set. This corresponds to mapping every term in $H$ to a term of the form $\alpha_{ij}(-X_i X_j - Y_i Y_j + Z_i Z_j)$. This is a stoquastic matrix (i.e.\ all its off-diagonal entries are non-positive), so finding its ground-state energy is in $\stoqma$~\cite{bravyi06a}.
\end{proof}

See~\cite{piddock15} for subsequent work proving hardness for various interactions with mixed signs on square lattices, and antiferromagnetic interactions on triangular lattices (but not the Heisenberg interaction, which remains an open problem).


\subsection{The XY model}
\label{sec:xy}

Another case which we can explicitly solve, and will turn out to be similar, is the XY model. (More precisely, the special case of the XY model that is -- somewhat confusingly -- called the XX model.) Here we have access to terms of the form $H := XX + YY$. $\qma$-hardness of this model was proven by Biamonte and Love in the case where we also have access to arbitrary $X$ and $Y$ terms~\cite{biamonte08}. As with the Heisenberg model, we will encode logical qubits within triples of physical qubits $(1,2,3)$. It can be explicitly calculated that the ground space of $H_{12} + H_{13} + H_{23}$ is the asymmetric subspace (i.e.\ the orthogonal subspace to the symmetric subspace) of 3 qubits. Therefore, we can effectively project onto this subspace using Corollary~\ref{cor:zerothorder}. When restricted to this subspace, up to an overall multiplicative scaling, we have
%
\[ H_{12} = -2 ZI + I,\;\; H_{13} = \sqrt{3} XI + ZI -I,\;\; H_{23} = -\sqrt{3} XI + ZI -I . \]
Therefore, by taking linear combinations we can make $XI$ and $ZI$. When we apply $H$ interactions across pairs of triples, we get that all the matrices obtained are of the form $M^{(i,j)}_{13} (XX + YY)_{24}$.
%
%
%
We will not write out all the matrices $M^{(i,j)}$ explicitly, but simply record that, similarly to the Heisenberg model, we have
\[ \frac{3}{4}\left( M^{(1,4)} - M^{(1,5)} - M^{(2,4)} + M^{(2,5)}\right) = XX, \]
\[ \frac{1}{4}\left( M^{(1,4)} + M^{(1,5)} - 2 M^{(1,6)} + M^{(2,4)} + M^{(2,5)} - 2M^{(2,6)} - 2M^{(3,4)} - 2 M^{(3,5)} + 4M^{(3,6)}\right) = ZZ, \]
\[ \sum_{i=1}^3 \sum_{j=4}^6 M^{(i,j)} = II. \]
By a similar argument to the last section, if $\ket{\psi}$ is a nondegenerate ground state of a Hamiltonian built from $XX+YY$ terms, we can effectively make any Hamiltonian of the form
\[ \widetilde{H} = \sum_{k=1}^n \alpha_k X_k + \beta_k Z_k + \sum_{i<j} (\gamma_{ij} X_i X_j + \delta_{ij} Z_i Z_j) \bracket{\psi}{(XX + YY)_{i'j'}}{\psi}. \]
We therefore now need a Hamiltonian built from $XX+YY$ terms which has a unique ground state $\ket{\psi}$ such that all two-qubit reduced density matrices $\rho$ satisfy
\[ \tr \rho(XX + YY) \neq 0, \]
and where this quantity can be computed efficiently. In the case of the XY model, this is easier than for the Heisenberg model: we can use the complete graph of interactions. We observe that $XX+YY$ commutes with $ZI + IZ$, so we can specify joint eigenvectors of $\sum_i Z_i$ and
\[ H_C := \sum_{i<j} X_i X_j + Y_i Y_j = 2 \sum_{i< j} (\ket{01}\bra{10} + \ket{10}\bra{01})_{ij}. \]
In particular, every eigenvector of $H_C$ can be taken to consist of terms of constant Hamming weight. $H_C$ can be viewed as the adjacency matrix of an undirected graph. Within each subspace of constant Hamming weight $k$, this is just the Johnson graph $J(n,k)$ whose vertices are $k$-subsets of $[n]$, where two vertices are connected if their subsets differ by 2 elements. As this is a regular graph, its principal eigenvector is the uniform superposition over all vectors of fixed Hamming weight. Finally, to determine the eigenvalues, we simply compute
\[ \left(2 \sum_{i< j} (\ket{01}\bra{10} + \ket{10}\bra{01})_{ij}\right)\left(\sum_{x,|x|=k} \ket{x} \right) = 2\sum_{x,|x|=k} \sum_{y,|y|=k,d(x,y)=2} \ket{y} = 2k(n-k) \sum_{x,|x|=k} \ket{x}. \]
Thus, when $n$ is even, the maximal eigenvalue is unique and equal to $n^2/4$. The overlap between the 2-qubit reduced states of the corresponding principal eigenvector $\ket{\psi}$ and $(XX+YY)_{i'j'}$ is constant and non-zero for all $i'$, $j'$. So, using Corollary~\ref{cor:zerothorder} to project onto the ground state of $-H_C$, we can effectively produce any Hamiltonian of the form
\[ \widetilde{H} = \sum_{k=1}^n \alpha_k X_k + \beta_k Z_k + \sum_{i<j} \gamma_{ij} X_i X_j + \delta_{ij} Z_i Z_j, \]
which suffices for $\qma$-completeness by Lemma~\ref{lem:xxandzzlf}. We have proven the following lemma.

\begin{replem}{lem:xy}
{\sc $\{XX+YY\}$-Hamiltonian} is $\qma$-complete.
\end{replem}


\subsection{The skew-symmetric case \texorpdfstring{$XZ-ZX$}{XZ-ZX}}
\label{sec:skew}

The final individual special case we need to consider is the Hamiltonian $H := XZ - ZX$, which we will show is also $\qma$-complete. The proof is based on a reduction from the XY model.

We will use three physical qubits to correspond to one logical qubit, and apply a strong interaction of the form $H_{12} + H_{23} + H_{31}$. Because $H$ is not symmetric under interchange of qubits, it now makes a difference in which direction we apply these interactions (i.e.\ $H_{12} \neq H_{21}$). By Corollary~\ref{cor:zerothorder}, the result is that the physical qubits are effectively projected into the ground space of this matrix. It can be explicitly calculated that this is 2-dimensional and spanned by
\begin{align*}
&\frac{1}{2\sqrt{3}} \left(-\ket{001} +2 \ket{010} -\sqrt{3}\ket{011} -\ket{100} +\sqrt{3}\ket{110} \right),\\
&\frac{1}{2\sqrt{3}} \left( -\sqrt{3}\ket{001}-\ket{011}+\sqrt{3}\ket{100}+2\ket{101}-\ket{110} \right).
\end{align*}
We now apply $H$ interactions across pairs of triples $(1,2,3)$, $(4,5,6)$. It turns out that restricted to the above subspace and with respect to the above basis,
\[ H_{16} - H_{15} = \frac{4}{3\sqrt{3}}\left(XX + ZZ\right) \]
on the logical qubits. As the XY model is $\qma$-complete by Lemma~\ref{lem:xy}, by relabelling $Z$ to $Y$ we have proven the following result.

\begin{replem}{lem:xzskew}
{\sc $\{XZ-ZX\}$-Hamiltonian} is $\qma$-complete.
\end{replem}


\section{More general models}
\label{sec:general}

We now resolve the complexity of more general types of Hamiltonians, by reducing from the previously discussed special cases.

\subsection{The case \texorpdfstring{$XX + \alpha YY$}{XX+aYY}}
\label{sec:xxayy}

Using Hamiltonians of the form $XX + \alpha YY$, where $\alpha \notin \{0,1\}$, we can effectively emulate the XY model. Consider a line on 3 qubits with an $H := XX + \alpha YY$ interaction across each connected pair of qubits, i.e.\ an overall Hamiltonian of the form $H_{12} + H_{23}$. One can calculate explicitly that this has eigenvalues $\{\pm2\sqrt{1+\alpha^2},0\}$, and that the non-zero eigenvalues correspond to 2-dimensional subspaces. The lowest energy subspace is spanned by the states
\begin{align*}
& \left\{ \frac{1}{2}\sgn(\alpha-1) \ket{001} - \frac{(\alpha+1)\sgn(\alpha-1)}{2\sqrt{1+\alpha^2}}\ket{010} + \frac{1}{2}\sgn(\alpha-1) \ket{100} + \frac{|\alpha-1|}{2\sqrt{1+\alpha^2}} \ket{111}, \right. \\
&\left.\frac{\alpha-1}{2\sqrt{1+\alpha^2}}\ket{000} + \frac{1}{2}\ket{011} - \frac{1+\alpha}{2\sqrt{1+\alpha^2}} \ket{101} + \frac{1}{2} \ket{110} \right\}.
\end{align*}
We take two copies of this gadget, labelled with triples $(1,2,3)$ and $(4,5,6)$, and effectively project into this lowest energy subspace using Corollary~\ref{cor:zerothorder} to obtain a logical 2-qubit space. By applying $H$ interactions across different pairs of physical qubits we can make effective interactions
\[ H_{14} \mapsto \frac{1}{1+\alpha^2} XX + \frac{\alpha^3}{1+\alpha^2} YY,\;\;\;\; H_{24} \mapsto - \frac{1}{(1+\alpha^2)^{3/2}} XX - \frac{\alpha^4}{(1+\alpha^2)^{3/2}}YY, \]
where we use ``$\mapsto$'' to denote a physical interaction implying an effective interaction. Therefore, for any $\alpha\notin \{0,1\}$, by taking linear combinations of these interactions we can make the individual interactions $XX$ and $YY$ on the logical qubits, so we can make $XX+YY$ and hence get $\qma$-completeness.


\subsection{The case \texorpdfstring{$XX + \alpha YY + \beta ZZ$}{XX+aYY+bZZ}}
\label{sec:xxayybzz}

Set $H = XX + \alpha YY + \beta ZZ$, where at least one of $\alpha$ or $\beta$ is non-zero, and consider the matrix $M = H_{12} - H_{23}$. The eigenvalues of $M$ are in the set $\{0,\pm2\sqrt{1+\alpha^2+\beta^2}\}$, where the first has multiplicity 4 and the others multiplicity 2. This again lets us build one logical qubit out of three physical ones. We now take two copies of this gadget on qubits $(1,2,3)$ and $(4,5,6)$ and consider the interactions $(H_{24} + H_{35})/2$ and $H_{25}$. Projecting onto two copies of the ground space of $M$ and rescaling both interactions by $(1+\alpha^2+\beta^2)^2$, we get effective interactions
\[ \frac{1}{2}(H_{24} + H_{35}) \mapsto \alpha \beta XX + \alpha^3 \beta YY + \alpha \beta^3 ZZ =: H_1, \]
\[ H_{25} \mapsto XX + \alpha^5 YY + \beta^5 ZZ =: H_2. \]
In order to prove $\qma$-completeness, by the argument in Subsection~\ref{sec:xxayy} it suffices to produce an interaction which is a linear combination of these two matrices and has two non-zero components. So consider the matrix
\[ \alpha \beta H_2 - H_1 = \alpha^3(\alpha^3-1)\beta YY + \alpha \beta^3 (\beta^3-1)ZZ. \]
If $\alpha,\beta \notin \{0,1\}$, both the $YY$ and $ZZ$ parts are non-zero and $\qma$-completeness follows from relabelling Pauli matrices and the results of Subsection~\ref{sec:xxayy}. The only remaing case we have to prove is where $\alpha=1$ and $\beta \notin \{0,1\}$. In this special case, $H_1$ is proportional to $XX + YY + \beta^2 ZZ$ and $H_2 = XX + YY + \beta^5 ZZ$. So rescaling $H_1$ appropriately, $\beta^3 H_2 - H_1$ is proportional to $XX + YY$ and $\qma$-completeness follows from Lemma~\ref{lem:xy}.

Together with the previous subsection, we have proven the following lemma.

\begin{replem}{lem:xyz}
For any real $\beta$, $\gamma$ such that at least one of $\beta$ and $\gamma$ is non-zero, {\sc $\{XX+\beta YY+\gamma ZZ\}$-Hamiltonian} is $\qma$-complete.
\end{replem}


\section{Extracting local terms}
\label{sec:local}

The final $\qma$-complete cases we need to consider are those which contain 1\nobreakdash-local terms. There are a number of different cases which we will prove in turn, starting with the most ``generic'' symmetric case. Let $H$ be a two-qubit Hamiltonian of the form
\[ H = \alpha XX + \beta YY + \gamma ZZ + AI + IA, \]
where at least two of $\alpha$, $\beta$ and $\gamma$ are non-zero and $A$ is an arbitrary traceless Hermitian matrix. We will describe a gadget which allows us to implement the 1\nobreakdash-local interaction $A$. The gadget is of the following form, where $\Delta$ is some large positive number:
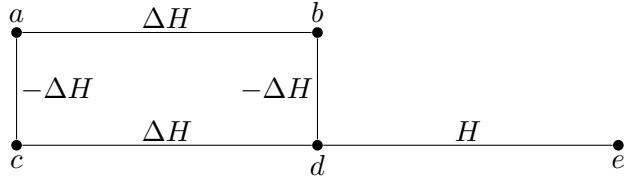
\begin{center}
\begin{tikzpicture}[inner sep=0.5mm,yscale=1.5,xscale=4]
\node (a) at (0,1) [circle,fill=black,label=north:$a$] {};
\node (b) at (1,1) [circle,fill=black,label=north:$b$] {};
\node (c) at (0,0) [circle,fill=black,label=south:$c$] {};
\node (d) at (1,0) [circle,fill=black,label=south:$d$] {};
\node (e) at (2,0) [circle,fill=black,label=south:$e$] {};
\draw (a) to node[midway,above] {$\Delta H$} (b);
\draw (a) to node[auto] {$-\Delta H$} (c);
\draw (b) to node[auto,swap] {$-\Delta H$} (d);
\draw (c) to node[midway,above] {$\Delta H$} (d);
\draw (d) to node[midway,above] {$H$} (e);
\end{tikzpicture}
\end{center}
Let $G := H_{ab} + H_{cd} - H_{ac} - H_{bd}$. We have
\[ G = \alpha(X_a - X_d)(X_b - X_c) + \beta(Y_a - Y_d)(Y_b - Y_c) + \gamma(Z_a - Z_d)(Z_b - Z_c); \]
observe that $A$ has disappeared. One can explicitly compute that the eigenvalues of $G$ are the set
\[ \{ 0, \pm 4\alpha, \pm 4 \beta, \pm 4 \gamma, \pm 4 \sqrt{\alpha^2 + \beta^2 + \gamma^2} \}. \]
As at least two of $\alpha$, $\beta$ and $\gamma$ are non-zero, the lowest eigenvalue is $-4 \sqrt{\alpha^2 + \beta^2 + \gamma^2}$, which is nondegenerate and corresponds to an (unnormalised) eigenvector
\begin{multline*}
(\alpha-\beta)(\ket{0000}+\ket{1111}) + (\gamma + \sqrt{\alpha^2+\beta^2+\gamma^2})(\ket{0101}+\ket{1010}) \\
- (\alpha+\beta)(\ket{0110}+\ket{1001}) + (\gamma - \sqrt{\alpha^2+\beta^2+\gamma^2})(\ket{0011}+\ket{1100}).
\end{multline*}
This eigenvector is maximally entangled across the split $(abc:d)$. Therefore, if we project $H_{de}$ onto the ground state $\ket{\psi}$ of $G$, the effective interaction produced is $A$ on qubit $e$, as the terms in $H$ which are not the identity on qubit $d$ vanish. Producing an $A$ term allows us to effectively delete the 1\nobreakdash-local part of $H$ and produce an interaction of the form $\alpha XX + \beta YY + \gamma ZZ$. Using Lemma~\ref{lem:xyz}, and potentially relabelling Paulis, we have proven the following lemma.

\begin{replem}{lem:extractlts}
For any $\beta$, $\gamma$ such that at least one of $\beta$ and $\gamma$ is non-zero, and any single-qubit Hermitian matrix $A$, {\sc $\{XX+\beta YY+\gamma ZZ + AI + IA\}$-Hamiltonian} is $\qma$-complete.
\end{replem}

Similarly, if $H$ is a two-qubit Hamiltonian of the form
\[ H = XZ - ZX + AI - IA, \]
where $A$ is an arbitrary traceless Hermitian matrix, we can extract $A$ using the gadget
\[ G = H_{ab} + H_{bc} + H_{cd} + H_{da}. \]
Note that in this case the direction in which we apply $H$ is important, as $H$ is not symmetric. One can again observe that $A$ disappears from $G$, and that $G$ has a unique lowest eigenvalue $-4 \sqrt{2}$. This again corresponds to an eigenvector which is maximally entangled across the split $(abc:d)$, allowing us to effectively produce an interaction $A$ in the same way as before, and implying (by Lemma~\ref{lem:xzskew}) the following.

\begin{replem}{lem:extractltsskew}
For any single-qubit Hermitian matrix $A$, {\sc $\{XZ-ZX + AI - IA\}$-Hamiltonian} is $\qma$-complete.
\end{replem}

Next, if $H$ is of the form
\[ H = XX + \alpha(XI + IX) + \beta(ZI + IZ), \]
we use the gadget $G = H_{ab} - H_{bc}$. It can be explicitly calculated that this has lowest eigenvalue $-2\sqrt{(1\pm\alpha)^2 + \beta^2}$, depending on whether $\alpha$ is positive or negative, which is unique when $\alpha \neq 0$. The corresponding eigenstate $\ket{\psi}$ has reduced density matrices
\[ \rho_b = (I \pm X) / 2,\;\;\;\; \rho_c = \frac{1}{2} I + \frac{\pm1 + \alpha}{2\sqrt{(1 \pm \alpha)^2 + \beta^2}} X +  \frac{\beta}{2\sqrt{(\pm1 + \alpha)^2 + \beta^2}} Z. \]
Therefore, if we project qubits $a$, $b$, $c$ onto $\ket{\psi}$ using the configuration
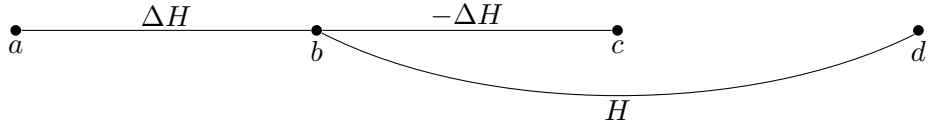
\begin{center}
\begin{tikzpicture}[inner sep=0.5mm,yscale=2,xscale=4]
\node (a) at (0,0) [circle,fill=black,label=south:$a$] {};
\node (b) at (1,0) [circle,fill=black,label=south:$b$] {};
\node (c) at (2,0) [circle,fill=black,label=south:$c$] {};
\node (d) at (3,0) [circle,fill=black,label=south:$d$] {};
\draw (a) to node[midway,above] {$\Delta H$} (b);
\draw (b) to node[auto] {$-\Delta H$} (c);
\draw[out=315,in=225] (b) to node[midway,below] {$H$} (d);
\end{tikzpicture}
\end{center}
or
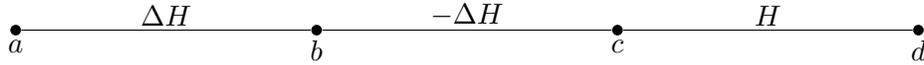
\begin{center}
\begin{tikzpicture}[inner sep=0.5mm,yscale=2,xscale=4]
\node (a) at (0,0) [circle,fill=black,label=south:$a$] {};
\node (b) at (1,0) [circle,fill=black,label=south:$b$] {};
\node (c) at (2,0) [circle,fill=black,label=south:$c$] {};
\node (d) at (3,0) [circle,fill=black,label=south:$d$] {};
\draw (a) to node[midway,above] {$\Delta H$} (b);
\draw (b) to node[auto] {$-\Delta H$} (c);
\draw (c) to node[midway,above] {$H$} (d);
\end{tikzpicture}
\end{center}
explicit calculation shows that we implement 1\nobreakdash-local interactions on qubit $d$ proportional to
\[ \pm I + (\pm1+\alpha)X + \beta Z,\;\;\;\; (\alpha(\alpha \pm 1) + \beta^2) I + (\pm1 + \alpha + \alpha \sqrt{(\pm1 + \alpha)^2 + \beta^2})X + \beta \sqrt{(\pm1 + \alpha)^2 + \beta^2} Z. \]
By taking linear combinations and ignoring the irrelevant $I$ term we can therefore make arbitrary 1\nobreakdash-local interactions of the form $\alpha X + \beta Z$.

In the case where $\alpha = 0$, so we have
\[ H = XX + \beta (ZI + IZ), \]
we can use an even simpler gadget:
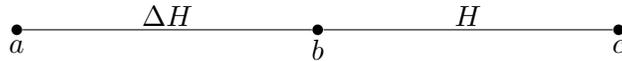
\begin{center}
\begin{tikzpicture}[inner sep=0.5mm,yscale=2,xscale=4]
\node (a) at (0,0) [circle,fill=black,label=south:$a$] {};
\node (b) at (1,0) [circle,fill=black,label=south:$b$] {};
\node (c) at (2,0) [circle,fill=black,label=south:$c$] {};
\draw (a) to node[midway,above] {$\Delta H$} (b);
\draw (b) to node[midway,above] {$H$} (c);
\end{tikzpicture}
\end{center}
When $\beta \neq 0$, $H$ has a unique lowest eigenvalue $-\sqrt{1+4\beta^2}$, corresponding to an eigenstate proportional to $(2\beta - \sqrt{1+4 \beta^2})\ket{00} + \ket{11}$. One can verify that this produces an interaction on qubit $c$ proportional to
\[ -\frac{2 \beta^2}{\sqrt{1+4\beta^2}} I + \alpha Z. \]
Therefore, we can implement arbitrary $Z$ and $XX$ interactions. This allows us effectively to implement $X$ interactions too, using a similar technique to the proof of Lemma~\ref{lem:zzhamlf}. We attach an additional ancilla qubit, and when we want to implement an $X$ on qubit $i$, we implement $XX$ on the qubit and the ancilla. The Hamiltonian we implement is therefore of the form
\[ H' = \sum_{i < j} \alpha_{ij} X_i X_j + X_{\text{anc}} \sum_k \beta_k X_k + \gamma_k Z_k =: \helse +  X_{\text{anc}} \sum_k \beta_k X_k. \]

Any vector $\ket{\psi}$ on $n+1$ qubits can be decomposed in terms of $\ket{+}$, $\ket{-}$ on the ancilla qubit, so
\begin{align*}
\bracket{\psi}{H'}{\psi} &= \left(\bra{+}\bra{\psi_+} + \bra{-}\bra{\psi_-}\right) \left(\helse + X_{\text{anc}} \sum_k \beta_k X_k\right) \left(\ket{+}\ket{\psi_+} + \ket{-}\ket{\psi_-} \right)\\
&= \bra{\psi_+} \left(\helse + \sum_k \beta_k X_k \right) \ket{\psi_+} + \bra{\psi_-} \left(\helse - \sum_k \beta_k X_k\right) \ket{\psi_-},
\end{align*}
where $\ket{\psi_+}$ and $\ket{\psi_-}$ are arbitrary. By convexity, the minimum is obtained by making at most one of these two terms non-zero. Now observe that the two matrices $\helse \pm \sum_k \beta_k X_k$ are actually unitarily equivalent, because by conjugating either of them with $Z^{\otimes n}$, we flip the sign of the local $X$ terms, but leave $\helse$ invariant. This latter property holds because the remainder of the Hamiltonian consists only of $XX$ and $Z$ terms, each of which is invariant under conjugation of this form. Thus the minimal eigenvalue of $H'$ is equal to the minimal eigenvalue of $H$, and we have essentially implemented an arbitary Hamiltonian of the form
\[ \sum_{i < j} \alpha_{ij} X_i X_j + \sum_k \beta_k X_k + \gamma_k Z_k. \]
Relabelling $X$ and $Z$, we have proven the following lemma.

\begin{replem}{lem:extractlts2}
For any single-qubit Hermitian matrix $A$ such that $A$ does not commute with $Z$, {\sc $\{ZZ, X, Z\}$-Hamiltonian} reduces to {\sc $\{ZZ+ AI + IA\}$-Hamiltonian}.
\end{replem}

The final case we consider is a similar, skew-symmetric version: now we have access to the pair $H = XX$, $H' = \alpha(XI-IX) + \beta(ZI-IZ)$, where $\beta \neq 0$. $H'$ has a nondegenerate ground state which is proportional to
\[ -\alpha \ket{00} + (\beta - \sqrt{\alpha^2 + \beta^2})\ket{01} +(\beta + \sqrt{\alpha^2 + \beta^2})\ket{10} + \alpha\ket{11}. \]
Using the gadget
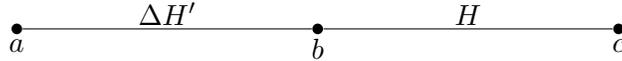
\begin{center}
\begin{tikzpicture}[inner sep=0.5mm,yscale=2,xscale=4]
\node (a) at (0,0) [circle,fill=black,label=south:$a$] {};
\node (b) at (1,0) [circle,fill=black,label=south:$b$] {};
\node (c) at (2,0) [circle,fill=black,label=south:$c$] {};
\draw (a) to node[midway,above] {$\Delta H'$} (b);
\draw (b) to node[midway,above] {$H$} (c);
\end{tikzpicture}
\end{center}
we get an effective interaction on $c$ which is non-zero for any $\beta \neq 0$, and is proportional to $Z$. Therefore, we can make $\{Z,XX\}$, which by the previous argument suffices to implement $X$ too. We encapsulate this as the following lemma, in which we relabel $X$ and $Z$.

\begin{replem}{lem:extractlts3}
For any single-qubit Hermitian matrix $A$ such that $A$ does not commute with $Z$, {\sc $\{ZZ, X, Z\}$-Hamiltonian} reduces to {\sc $\{ZZ, AI-IA\}$-Hamiltonian}.
\end{replem}


\section{The diagonal case}
\label{sec:diagonal}

The very last case we need to deal with is diagonal 2-qubit matrices. This setting is similar to previous work by Creignou~\cite{creignou95} and Khanna, Sudan and Williamson~\cite{khanna97} proving hardness of weighted variants of the {\sc Max-CSP} problem, but is not quite the same as here we allow both positive and negative weights. Indeed, cases which are in $\ptime$ in the model of~\cite{creignou95,khanna97} turn out to be $\nptime$-complete here. This setting was previously considered by Jonsson~\cite{jonsson00}, who completely classified the complexity of maximisation variants of boolean constraint satisfaction problems with arbitrary positive or negative weights. The following lemma is a special case of his classification result; we include a simple direct proof.

\begin{replem}{lem:np}
Let $\mathcal{S}$ be a set of diagonal Hermitian matrices on at most 2 qubits. Then, if every matrix in $\mathcal{S}$ is 1\nobreakdash-local, {\sc $\mathcal{S}$-Hamiltonian} is in $\ptime$. Otherwise, {\sc $\mathcal{S}$-Hamiltonian} is $\nptime$-complete.
\end{replem}

\begin{proof}
The first claim is obvious. For the second claim, we first note that {\sc $\mathcal{S}$-Hamiltonian} is in $\nptime$, as the minimal eigenvalue is achieved on a computational basis state. Assume that $\mathcal{S}$ contains a matrix $H$ which is not 1\nobreakdash-local, and write $H = \diag(\alpha,\beta,\gamma,\delta)$. To show that {\sc $\mathcal{S}$-Hamiltonian} is $\nptime$-hard, we split into cases.

First, assume that $\alpha$ is the unique minimum value on the diagonal of $H$. Given two ancilla qubits $c$, $d$, applying a heavily weighted $H$ term across them forces them both to be in the $\ket{0}$ state. Then a (normally-weighted) $H_{cb}$ interaction effectively implements a $\diag(\alpha,\beta)$ term on the $b$ qubit; similarly, an $H_{ac}$ interaction effectively implements $\diag(\alpha,\gamma)$ on the $a$ qubit. Considering both qubits together, this implies that we can effectively implement the 2-qubit terms
\[ \{\diag(\alpha,\beta,\gamma,\delta),\diag(\alpha,\beta,\alpha,\beta),\diag(\alpha,\alpha,\gamma,\gamma),\diag(1,1,1,1)\}, \]
the last of which is just the identity. Can we build general diagonal matrices as linear combinations of these terms? Observe that the matrix
\[ \begin{pmatrix} \alpha & \alpha & \alpha & 1\\ \beta & \beta & \alpha & 1\\ \gamma & \alpha & \gamma & 1\\ \delta & \beta & \gamma & 1 \end{pmatrix} \]
is full rank unless either $\alpha=\beta$, $\alpha=\gamma$ (both of which are disallowed by our assumption that $\alpha$ is the unique minimum), or $\alpha + \delta = \beta + \gamma$. This last possibility would imply that $H$ is 1\nobreakdash-local, so is also disallowed. So, as this matrix is full rank, any diagonal matrix can be implemented, and in particular the matrix $ZZ$, implying that the $\nptime$-complete problem MAX-CUT reduces to {\sc $\mathcal{S}$-Hamiltonian}. A similar argument goes through when each of the other diagonal entries of $H$ is the unique minimum; by negating $H$, the same applies when $H$ has a unique maximum.

We are left with the case that there is a non-unique maximum and a non-unique minimum on the diagonal of $H$. If $H$ is of the form $\diag(\alpha,\beta,\alpha,\beta)$ or $\diag(\alpha,\alpha,\beta,\beta)$, it is 1\nobreakdash-local. So $H$ must be of the form $\diag(\alpha,\beta,\beta,\alpha)$, with $\alpha\neq\beta$. This implies that we can produce $ZZ$ by taking linear combinations of $H$ and the identity, so once again MAX-CUT reduces to {\sc $\mathcal{S}$-Hamiltonian}.
\end{proof}


\section{Outlook}
\label{sec:outlook}

We have completely resolved the complexity of a natural subclass of \sham\ problems. However, many interesting generalisations and open problems remain in this area, such as:
\begin{enumerate}
\item Can we generalise our results to $k$\nobreakdash-local Hamiltonians for $k>2$? Although we achieved this for \shamlf, one potentially significant difficulty with improving this to the full \sham\ problem is that we know of no suitable normal form~\cite{bennett02} for Hermitian matrices on $k \ge 3$ qubits. Another issue is that reduction to the 2\nobreakdash-local case, which we used for \shamlf, does not seem easy to perform without having access to 1\nobreakdash-local terms.

\item Can we generalise our results beyond qubits? Again, this could be difficult as the equivalent generalisation of Schaefer's dichotomy theorem~\cite{schaefer78} to constraint satisfaction problems on a 3-element domain took 24 years, being resolved in 2002 by Bulatov~\cite{bulatov02}.

\item Can we prove hardness, or otherwise, for more restricted types of Hamiltonian? One way of restricting further would be to put limitations on the signs or types of coefficients allowed (such as the antiferromagnetic Heisenberg model), another would be to restrict the geometry of interactions (such as only allowing a planar graph, or a square lattice). Subsequent work by Piddock and one of us~\cite{piddock15} has achieved this in some cases, by showing that the antiferromagnetic XY and Heisenberg interactions are $\qma$-complete, and that many of the interactions proven $\qma$-complete here are still $\qma$-complete on a square lattice. For some interactions, these sign and geometry restrictions can be combined; for example, the antiferromagnetic XY interaction on a triangular lattice is $\qma$-complete.

Another case which has been of interest is Hamiltonians whose terms commute pairwise. In this case the $k$\nobreakdash-local Hamiltonian problem is in $\ptime$ for various special cases: 2\nobreakdash-local Hamiltonians~\cite{bravyi05}, 3\nobreakdash-local qubit Hamiltonians~\cite{aharonov11}, and $k$\nobreakdash-local Hamiltonians whose terms are projectors onto eigenspaces of Pauli matrices~\cite{yan12}. One more example in this vein is a result of Schuch proving that the problem is in $\nptime$ for a special class of commuting 4\nobreakdash-local qubit Hamiltonians~\cite{schuch11}.


\item Our results can be seen as a quantum generalisation of dichotomy theorems for the {\sc $k$-Max-CSP} problem~\cite{creignou01}. Another way to generalise Schaefer's original dichotomy theorem~\cite{schaefer78} would be to prove a similar result for the quantum $k$-SAT problem. This is a variant of {\sc $k$\nobreakdash-local Hamiltonian} where each term is a projector, and we ask whether there exists a state which is in the nullspace of all the projectors (``satisfies all the constraints''). This problem was introduced by Bravyi~\cite{bravyi11}, who proved it is within $\ptime$ for $k=2$, and $\qma_1$-complete for $k \ge 4$, where $\qma_1$ is the 1-sided error variant of $\qma$. Gosset and Nagaj recently improved this result to prove that quantum 3-SAT is also $\qma_1$-complete~\cite{gosset13}.

\item In a different direction, an interesting open question is whether one can prove a dichotomy theorem for unitary quantum gates. For example, given a set $\mathcal{G}$ of unitary gates, are circuits made up of gates picked from $\mathcal{G}$ always either classically simulable or universal for $\bqptime$? This question was resolved quite recently for gates produced by applying  2\nobreakdash-local Hamiltonians from a given set for arbitrary lengths of time~\cite{childs11b}. The general question is likely to be sensitive to the precise definitions of ``simulable'' and ``universal'', as demonstrated by the apparently intermediate class of commuting quantum computations~\cite{bremner11}.
\end{enumerate}


\subsection*{Acknowledgements}

Some of this work was completed while the authors were at the University of Cambridge. AM is supported by the UK EPSRC under Early Career Fellowship EP/L021005/1. TC is supported by the Royal Society. AM would like to thank Mick Bremner for pointing out reference~\cite{kempe00}, and we would like to thank various referees for their helpful comments. Special thanks to Laura Man\v{c}inska for spotting an error in a previous version.


\appendix


\section{Complexity class definitions}
\label{app:classes}

The two quantum variants of Merlin-Arthur complexity classes which we use are formally defined here (see~\cite{watrous09,bravyi06a} for more details). They are defined in terms of {\em promise problems}, i.e.\ pairs $A_{\text{yes}}, A_{\text{no}} \subseteq \{0,1\}^*$ such that $A_{\text{yes}} \cap A_{\text{no}} = \emptyset$, where the answer to the problem should be ``yes'' for inputs in $A_{\text{yes}}$, and ``no'' for inputs in $A_{\text{no}}$.



\begin{dfn}[Quantum Merlin-Arthur~\cite{kitaev02}]
A promise problem $A = (A_{\text{yes}},A_{\text{no}})$ is in $\qma$ if and only if there exists a polynomially-bounded function $p$ and a uniformly-generated family of quantum circuits $\{C_n\}$ such that, for all $n$ and all $x \in \{0,1\}^n$:
\begin{itemize}
\item If $x \in A_{\text{yes}}$, there exists a $p(n)$-qubit quantum state $\ket{\psi}$ such that $\Pr[C_n \text{ accepts } (x,\ket{\psi})] \ge 2/3$;
\item If $x \in A_{\text{no}}$, for all $p(n)$-qubit quantum states $\ket{\psi}$, $\Pr[C_n \text{ accepts } (x,\ket{\psi})] \le 1/3$.
\end{itemize}
\end{dfn}

The second variant we define is possibly less familiar.

\begin{dfn}[Stoquastic Quantum Merlin-Arthur~\cite{bravyi06a}]
A stoquastic verifier $V_n$ is described by a tuple $(n,n_w,n_0,n_+,U)$, where $U$ is a quantum circuit on $n+n_w+n_0+n_+$ qubits consisting of X, CNOT and Toffoli gates. The acceptance probability of the stoquastic verifier $V_n$ on input string $x \in \{0,1\}^n$ and input state $\ket{\psi} \in (\C^2)^{\otimes n_w}$ is $\bracket{\psi_{\text{in}}}{U^\dag \Pi_{\text{out}} U}{\psi_{\text{in}}}$, where $\ket{\psi_{\text{in}}} = \ket{x}\ket{\psi}\ket{0}^{n_0}\ket{+}^{n_+}$ and $\Pi_{\text{out}} = \proj{+}_1$ is the measurement which projects the first qubit onto the state $\ket{+} = \frac{1}{\sqrt{2}}(\ket{0} + \ket{1})$.

A promise problem $A = (A_{\text{yes}},A_{\text{no}})$ is in $\stoqma$ if and only if there exists a polynomially-bounded function $p$ and a uniformly-generated family of stoquastic verifiers $\{V_n\}$, where $V_n$ has at most $p(n)$ gates, such that, for all $n$ and all $x \in \{0,1\}^n$:
\begin{itemize}
\item If $x \in A_{\text{yes}}$, there exists an $n_w$-qubit quantum state $\ket{\psi}$ such that $\Pr[V_n \text{ accepts } (x,\ket{\psi})] \ge c$;
\item If $x \in A_{\text{no}}$, for all $n_w$-qubit quantum states $\ket{\psi}$, $\Pr[V_n \text{ accepts } (x,\ket{\psi})] \le s$,
\end{itemize}
where $|c-s| \ge 1/p(n)$.
\end{dfn}


\section{Characterisations of diagonalisability by local unitaries}
\label{app:dlu}

In this appendix we describe two equivalent characterisations of the property of being simultaneously diagonalisable by local unitaries.

\begin{lem}
Let $\{H_n\}$ be a finite set of Hermitian 2-qubit matrices. Then the following are equivalent:
\item
\begin{enumerate}
\item There exists $U \in U(2)$ such that, for all $n$, $U^{\otimes 2} H_n (U^{-1})^{\otimes 2}$ is diagonal;
\item For all $n$, $H_n = \alpha_n A \otimes A + \beta_n A \otimes I + \gamma_n I \otimes A + \delta_n I\otimes I$, for some single-qubit Hermitian matrix $A$ and real coefficients $\alpha_n$, $\beta_n$, $\gamma_n$, $\delta_n$;
\item For all $m,n$ (including $m=n$), $[H_n \otimes I,I \otimes H_m]=[H_n \otimes I,I \otimes FH_mF]=[FH_nF \otimes I,I \otimes H_m]=0$.
\end{enumerate}
\end{lem}

\begin{proof}
\begin{itemize}
\item $(1) \Rightarrow (2)$: Assuming (1) holds, we have
\[ U^{\otimes 2} H_n (U^{-1})^{\otimes 2} = \alpha_n Z \otimes Z + \beta_n Z \otimes I + \gamma_n I \otimes Z + \delta_n I\otimes I \]
for real coefficients $\alpha_n$, $\beta_n$, $\gamma_n$, $\delta_n$, so we can take $A = U^{-1}Z U$ and $H_n$ is of the form (2).
\item $(2) \Rightarrow (3)$: Obvious by direct calculation.
\item $(3) \Rightarrow (1)$: Write
\[ H_n = \sum_{i=0}^3 A_i^{(n)} \otimes \sigma^i = \sum_{i=0}^3 \sigma^i \otimes B_i^{(n)} \]
for some Hermitian matrices $A_i^{(n)}$, $B_i^{(n)}$. Then the constraint that $[H_n \otimes I,I \otimes H_m]=0$ implies that
\begin{align*}
(H_n \otimes I)(I \otimes H_m) &= \left(\sum_{i=0}^3 \sigma^i \otimes B_i^{(n)} \otimes I\right) \left(\sum_{j=0}^3 I \otimes A_j^{(m)} \otimes \sigma^j\right) = \sum_{i,j=0}^3 \sigma^i \otimes (B_i^{(n)} A_j^{(m)}) \otimes \sigma^j\\
&=  \sum_{i,j=0}^3 \sigma^i \otimes (A_j^{(m)} B_i^{(n)}) \otimes \sigma^j = (I \otimes H_m)(H_n \otimes I),
\end{align*}
so $[B_i^{(n)},A_j^{(m)}]=0$ for all pairs $(i,j)$ and $(n,m)$. Similarly,
\begin{align*}
(H_n \otimes I)(I \otimes FH_mF) &= \left(\sum_{i=0}^3 \sigma^i \otimes B_i^{(n)} \otimes I\right) \left(\sum_{j=0}^3 I \otimes B_j^{(m)} \otimes \sigma^j\right) = \sum_{i,j=0}^3 \sigma^i \otimes (B_i^{(n)} B_j^{(m)}) \otimes \sigma^j\\
&=  \sum_{i,j=0}^3 \sigma^i \otimes (B_j^{(n)} B_i^{(m)}) \otimes \sigma^j = (I \otimes FH_mF)(H_n \otimes I),
\end{align*}
so $[B_i^{(n)},B_j^{(m)}]=0$ for all pairs $(i,j)$ and $(n,m)$. Finally, by a similar argument $[FH_nF \otimes I,I \otimes H_m]=0$ implies that $[A_i^{(n)},A_j^{(m)}]=0$ for all pairs $(i,j)$ and $(n,m)$. Together, we have that all of the matrices $A_i^{(n)}$, $B_i^{(m)}$ commute pairwise, implying that there exists $U$ such that $UA_i^{(n)}U^{-1}$ and $UB_i^{(n)}U^{-1}$ are diagonal for all $i$ and $n$. So $(I \otimes U) H_n (I \otimes U^{-1})$ only has $I$ and $Z$ terms appearing on the second subsystem in its Pauli expansion, and as
\begin{align*}
U^{\otimes 2} H_n (U^{-1})^{\otimes 2} &= (U \otimes I)(I \otimes U) H_n (I \otimes U^{-1}) (U^{-1} \otimes I)\\
&= (U \otimes I) \left( M_1^{(n)} \otimes I + M_2^{(n)} \otimes Z \right) (U^{-1} \otimes I)\\
&= U M_1^{(n)} U^{-1} \otimes I + U M_2^{(n)} U^{-1} \otimes Z,
\end{align*}
for some matrices $M_1^{(n)}$, $M_2^{(n)}$ which are linear combinations of the $A_i^{(n)}$ matrices, we have that $U^{\otimes 2} H_n (U^{-1})^{\otimes 2}$ is diagonal.
\end{itemize}
\end{proof}

The same idea works for arbitrary $k$, though in part (3) we need the matrices to commute for any choice of a single qubit where they hit each other; we omit the proof. Also observe that characterisation (3) gives an efficient test for whether $\{H_n\}$ is simultaneously diagonalisable by a local unitary. Note that the diagonalising unitary $U$ is never unique: it suffices for $U$ to diagonalise the matrix $A$ in characterisation (2), and it is easy to see that if $U$ diagonalises $A$, then so does the matrix $XU$.

Yet another equivalent characterisation of this property\footnote{We would like to thank an anonymous referee for pointing this out.} is that a 2-qubit matrix $H$ is locally diagonalisable if and only if there exists a single-qubit matrix $A$, not proportional to $I$, such that $[A \otimes I, H] = [H, I \otimes A] = 0$. As the existence of such a matrix can be tested by solving a system of linear equations, this characterisation gives an alternative efficient test for local diagonalisability.


\section{Normal form for Hermitian matrices}
\label{app:normal}

In this appendix we prove the lemmas stated in Section~\ref{sec:normal}.

\begin{replem}{lem:rotate}
Let $H$ be a traceless 2-qubit Hermitian matrix and write
\[ H = \sum_{i,j=1}^3 M_{ij} \sigma^i \otimes \sigma^j + \sum_{k=1}^3 v_k \sigma^k \otimes I + w_k I \otimes \sigma^k. \]
Then, for any orthogonal matrix $R \in SO(3)$, there exists $U \in SU(2)$ such that
\[ U^{\otimes 2} H (U^\dag)^{\otimes 2} = \sum_{i,j=1}^3 (RMR^T)_{ij} \sigma^i \otimes \sigma^j + \sum_{k=1}^3 (Rv)_k \sigma^k \otimes I + (Rw)_k I \otimes \sigma^k. \]
\end{replem}

\begin{proof}
For each $R \in SO(3)$ we use the homomorphism between $SU(2)$ and $SO(3)$ to associate $R$ with $U \in SU(2)$ such that $U \sigma^i U^{\dag} = \sum_j R _{ji} \sigma^j$. Then
\begin{align*}
U^{\otimes 2} &H (U^\dag)^{\otimes 2} \\
 &= \sum_{i,j=1}^3 M_{ij} ( U \sigma^i U^{\dag}) \otimes ( U \sigma^j U^{\dag}) + \sum_{k=1}^3 v_k (U \sigma^k U^{\dag}) \otimes I + w_k I \otimes (U \sigma^k U^{\dag})\\
&= \sum_{i,j=1}^3 M_{ij} \left(\sum_{\ell=1}^3 R_{\ell i} \sigma^{\ell}\right) \otimes \left(\sum_{m=1}^3 R_{mj} \sigma^m\right) + \sum_{k=1}^3 v_k \left(\sum_{\ell=1}^3 R_{\ell k} \sigma^{\ell}\right) \otimes I + w_k I \otimes \left(\sum_{\ell=1}^3 R_{\ell k} \sigma^{\ell}\right)\\
&= \sum_{\ell,m=1}^3 \left(\sum_{i,j=1}^3 R_{\ell i} M_{ij} R_{mj} \right) \sigma^{\ell} \otimes \sigma^m + \sum_{\ell=1}^3 \left(\sum_{k=1}^3 R_{\ell k} v_k \right)\sigma^\ell \otimes I + \left( \sum_{k=1}^3 R_{\ell k} w_k \right)I \otimes \sigma^\ell \\
&= \sum_{i,j=1}^3 (RMR^T)_{ij} \sigma^i \otimes \sigma^j + \sum_{k=1}^3 (Rv)_k \sigma^k \otimes I + (Rw)_k I \otimes \sigma^k.
\end{align*}
\end{proof}

\begin{replem}{lem:normalform}
Let $H$ be a traceless 2-qubit Hermitian matrix. If $H$ is symmetric under exchanging the two qubits on which it acts, there exists $U \in SU(2)$ such that
\[ U^{\otimes 2} H (U^\dag)^{\otimes 2} = \sum_{i=1}^3 \alpha_i \sigma^i \otimes \sigma^i + \sum_{j=1}^3 \beta_j (\sigma^j \otimes I + I \otimes \sigma^j), \]
for some real coefficients $\alpha_i$, $\beta_j$. If $H$ is antisymmetric under this exchange, there exists $U \in SU(2)$ and $i \neq j$ such that
\[ U^{\otimes 2} H (U^\dag)^{\otimes 2} = \alpha (\sigma^i \otimes \sigma^j - \sigma^j \otimes \sigma^i) + \sum_{k=1}^3 \beta_k (\sigma^k \otimes I - I \otimes \sigma^k), \]
for some real coefficients $\alpha$, $\beta_k$.
\end{replem}

\begin{proof}
In the first case, if $H$ has this symmetry then it must be of the form
\[ H = \sum_{i,j=1}^3 M_{ij} \sigma^i \otimes \sigma^j + \sum_{k=1}^3 \beta_k (\sigma^k \otimes I + I \otimes \sigma^k), \]
where $M$ is a symmetric matrix. Thus $M$ can be diagonalised by an orthogonal matrix $O$, which implies that it is diagonalisable by a special orthogonal matrix $S$. Indeed, the former condition is equivalent to saying that the columns of $O$ are a basis of eigenvectors of $M(H)$; if $\det O = -1$, then we can simply permute the columns of $O$ to make $\det O=1$. By Lemma~\ref{lem:rotate}, this corresponds to conjugating $H$ by two copies of some unitary $U \in SU(2)$. In the case where $H$ is antisymmetric under exchanging the two qubits, it must be of the form
\[ H = \sum_{i,j=1}^3 M_{ij} \sigma^i \otimes \sigma^j + \sum_{k=1}^3 \beta_k (\sigma^k \otimes I - I \otimes \sigma^k), \]
where $M$ is skew-symmetric. By conjugating by an orthogonal matrix $O$ we can map $M$ to a matrix of the form
\[
\begin{pmatrix}
0 & -\alpha & 0\\
\alpha & 0 & 0\\
0 & 0 & 0
\end{pmatrix}
\]
for some real $\alpha$ (see e.g.~\cite{thompson88} for a proof). By permuting rows and columns, we can assume that $O \in SO(3)$ as before.
\end{proof}


\section{\texorpdfstring{$\qma$}{QMA}-hardness of special cases of \texorpdfstring{\shamlf}{S\nobreakdash-local Hamiltonian with local terms}}
\label{sec:qmalf}

In this appendix we prove $\qma$-hardness of the various special cases of \shamlf\ which are required to complete the proof of Proposition~\ref{prop:2shamlf}. The second-order perturbation theory that we need to use can be encapsulated as the following lemma. For example, this lemma allows us to generate an effective interaction of the form $A\otimes D$ by using interactions of the form $A\otimes B$ and $C\otimes D$.

\begin{lem}
\label{lem:gadget}
Consider Hamiltonians $H^{(1)} = \sum_i A^{(i)} \otimes B^{(i)}$,  $H^{(2)} = \sum_i C^{(i)} \otimes D^{(i)}$, and two orthogonal states $\ket{\psi},\ket{\psi^\perp} \in B(\C^2)$. Take an overall Hamiltonian $\helse$ such that $\|\helse\| = \Omega(1)$ and pick two qubits $a$, $c$ from those on which $\helse$ acts. Add an extra qubit $b$, and set
\[ \widetilde{H} = \helse + \sqrt{\Delta} H^{(1)}_{ab} + \sqrt{\Delta} H^{(2)}_{bc} + \Delta \proj{\psi}_b + L_a + M_c, \]
where $L$ and $M$ are local terms to be determined, and $\Delta = \delta^2 \|\helse\|^2$, for arbitrary $\delta>1$. This is illustrated by the following diagram:
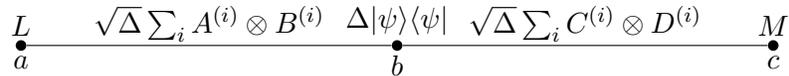
\begin{center}
\begin{tikzpicture}[inner sep=0.5mm,xscale=5]
\node (a) at (0,0) [circle,fill=black,label=south:$a$,label=north:$L$] {};
\node (b) at (1,0) [circle,fill=black,label=south:$b$,label=north:$\Delta \proj{\psi}$] {};
\node (c) at (2,0) [circle,fill=black,label=south:$c$,label=north:$M$] {};
\draw (a) to node[midway,above] {$\sqrt{\Delta} \sum_i A^{(i)} \otimes B^{(i)}$} (b); \draw (b) to node[midway,above] {$\sqrt{\Delta} \sum_i C^{(i)} \otimes D^{(i)}$} (c);
\end{tikzpicture}
\end{center}
Define
\bes
\heff' = H_{\operatorname{else}} - 2 \sum_{i,j} \Re \left(\bracket{\psi^\perp}{B^{(i)}}{\psi} \bracket{\psi}{C^{(j)}}{\psi^\perp} \right) A^{(i)}_a D^{(j)}_c.
\ees
Then there exist efficiently computable $L$ and $M$ such that
\[ \|\widetilde{H}_{<\Delta/2} - \heff' \proj{\psi^\perp}_b\| = O(\delta^{-1}). \]
\end{lem}

\begin{proof}
In the language of Corollary~\ref{cor:perturb}, we have $H = \Delta \proj{\psi}_b$, $V = H_{\operatorname{else}} + \sqrt{\Delta} H^{(1)}_{ab} + \sqrt{\Delta} H^{(2)}_{bc} + L_a + M_c$, so
\[ V_- = \left( H_{\operatorname{else}} + L_a + M_c + \sqrt{\Delta} \sum_i \bracket{\psi^\perp}{B^{(i)}}{\psi^\perp} A^{(i)}_a + \bracket{\psi^\perp}{C^{(i)}}{\psi^\perp} D^{(i)}_c \right) \proj{\psi^\perp}_b, \]
\[ V_{-+} = \left( \sqrt{\Delta} \sum_i \bracket{\psi^\perp}{B^{(i)}}{\psi} A^{(i)}_a + \bracket{\psi^\perp}{C^{(i)}}{\psi} D^{(i)}_c \right) \ket{\psi^\perp}\bra{\psi}_b, \]
\[ V_{+-} = \left( \sqrt{\Delta} \sum_i \bracket{\psi}{B^{(i)}}{\psi^\perp} A^{(i)}_a + \bracket{\psi}{C^{(i)}}{\psi^\perp} D^{(i)}_c \right) \ket{\psi}\bra{\psi^\perp}_b, \]
and hence
\begin{align*} V_{-+}V_{+-} = \Delta \Big(& \sum_{i,j} \bracket{\psi^\perp}{B^{(i)}}{\psi} \bracket{\psi}{B^{(j)}}{\psi^\perp} (A^{(i)} A^{(j)})_a \\
&+ \sum_{i,j} \bracket{\psi^\perp}{C^{(i)}}{\psi} \bracket{\psi}{C^{(j)}}{\psi^\perp} (D^{(i)} D^{(j)})_c\\
&+ 2 \sum_{i,j} \Re\left( \bracket{\psi^\perp}{B^{(i)}}{\psi} \bracket{\psi}{C^{(j)}}{\psi^\perp} \right) A^{(i)}_a D^{(j)}_c \Big) \proj{\psi^\perp}_b.
\end{align*}
%
Thus, by (\ref{eq:series2}), for some 1\nobreakdash-local term $L'$ on the $a$, $c$ systems (which depends on $\Delta$ and $z$),
\begin{align*} \Sigma_-(z) =& \Bigg( H_{\operatorname{else}} + L'_{ac} + L_a + M_c + \frac{2\Delta}{z-\Delta} \sum_{i,j} \Re\left( \bracket{\psi^\perp}{B^{(i)}}{\psi} \bracket{\psi}{C^{(j)}}{\psi^\perp} \right) A^{(i)}_a D^{(j)}_c \Bigg) \proj{\psi^\perp}_b\\
&+ O\left (\frac{\|V\|^3}{(z-\Delta)^2}\right).
\end{align*}
We now pick $L$ and $M$ such that $L_a + M_c = -L'_{ac}$ for $z=0$. Observe that $\|L\|, \|M\| = O(\sqrt{\Delta})$.  Identifying the first term with $\heff$ in Corollary~\ref{cor:perturb} and fixing $z=0$, we have $\lambda_+ = \Delta$, $\|\heff\| = O(\|\helse\|)$, $\|V\| = O(\|\helse\| + \sqrt{\Delta})$. Write
\[ \heff' = H_{\operatorname{else}} - 2 \sum_{i,j} \Re\left( \bracket{\psi^\perp}{B^{(i)}}{\psi} \bracket{\psi}{C^{(j)}}{\psi^\perp} \right) A^{(i)}_a D^{(j)}_c. \]
Then, by Corollary~\ref{cor:perturb},
%
%
\[ \|\widetilde{H}_{<\Delta/2} - \heff'\proj{\psi^\perp}_b\| = O\left( \frac{\|\helse\|(\|\helse\|+\sqrt{\Delta})}{\Delta} + \frac{(\|\helse\|+ \sqrt{\Delta})^3}{\Delta^2} \right). \]
Taking $\Delta = \delta^2 \|\helse\|^2$ for some $\delta>1$, we get $\|\widetilde{H}_{<\Delta/2} - \heff' \proj{\psi^\perp}_b\| = O(\delta^{-1})$.
\end{proof}

In particular, because of the simple product form of $\heff' \proj{\psi}_b^\perp$, the lowest eigenvalue of $\widetilde{H}$ is approximately equal to the lowest eigenvalue of $\heff'$. Observe that Lemma~\ref{lem:gadget} can be applied in series, but only a constant number of times, as each use of the lemma increases the norm of the Hamiltonian by a polynomial factor. However, the lemma can also be applied in parallel, i.e.\ different gadgets can be applied across an arbitrary number of distinct pairs of qubits, without changing the parameters at all.

Roughly speaking, this follows because the gadgets do not interfere with each other (to second order). This is justified more formally and generally in, for example,~\cite{oliveira08,bravyi08,cao14}; one can see it for the gadget used here as follows~\cite{piddock15}. Assume we are applying $k$ gadgets in parallel. Write the 1\nobreakdash-local term $\Delta \proj{\psi}$ within the $i$'th gadget as $H^{(i)}$ and the remaining terms in that gadget as $V^{(i)}$. Then the whole Hamiltonian can be written as $H+V$, where $H = \sum_i H^{(i)}$, $V = \sum_i V^{(i)}$. The effective Hamiltonian acts only on the ground space of $H$. Any state in this ground space can be written in the form $\ket{\psi_1}\dots\ket{\psi_k}\ket{\phi}$, where $\ket{\psi_i}$ is the ground state of $H^{(i)}$ (considered as a single-qubit matrix). To determine $\Sigma_-(z)$, in order to apply Corollary \ref{cor:perturb}, we need to compute $V_{-+} V_{+-}$ (see (\ref{eq:series})). As each $V^{(j)}$ only acts on the $j$'th mediator qubit and the non-mediator qubits, the $i$'th mediator qudit remains in the state $\ket{\psi_i}$ following the action of $V^{(j)}_{+-}$, implying that
\[ V_{-+} V_{+-} = \sum_{i,j} V^{(i)}_{-+} V^{(j)}_{+-} = \sum_i V^{(i)}_{-+} V^{(i)}_{+-}. \]
So the effective Hamiltonian simulated by $H+V$ is just the sum of the effective Hamiltonians simulated by $H^{(i)}$, $V^{(i)}$ separately.

Lemma~\ref{lem:gadget} allows us to prove $\qma$-hardness of a number of special cases of \shamlf. Given access to some set of interactions $\mathcal{S}$, we use the lemma to build a Hamiltonian containing additional interactions (up to an additive error $O(\delta^{-1})$, which we ignore for readability in what follows). If the new set of interactions $\mathcal{S}'$ corresponds to a $\qma$-complete problem {\sc $\mathcal{S}'$-Hamiltonian with local terms}, and we take $\delta = \poly(n)$, this implies that \shamlf\ is $\qma$-complete.

It was proven by Biamonte and Love~\cite{biamonte08} that {\sc 2\nobreakdash-local Hamiltonian} remains $\qma$-complete if the Hamiltonian only contains terms of the form $X$, $Z$, $XX$, $ZZ$. On the other hand, by giving a reduction from general 2\nobreakdash-local Hamiltonians, Oliveira and Terhal~\cite{oliveira08} showed that {\sc 2\nobreakdash-local Hamiltonian} is $\qma$-complete if the 2-body interactions are proportional to products of Pauli matrices and are restricted to the edges of a 2d square lattice. Their construction does not use any Y terms if they were not present already, so combining these results we get the following theorem.

\begin{thm}[Combination of Biamonte-Love~\cite{biamonte08} and Oliveira-Terhal~\cite{oliveira08}]
\label{thm:lattice}
{\sc 2\nobreakdash-local Hamiltonian} is $\qma$-complete, even if the Hamiltonian is of the form
\[ H = \sum_{(i,j) \in E} \alpha_{ij} A_i B_j + \sum_k C_k, \]
where $E$ is the set of edges of a 2-dimensional square lattice, $\alpha_{ij}$ are arbitrary real coefficients, each matrix $A_i$, $B_j$ is either $X$ or $Z$, and $C_k$ are arbitrary single-qubit Hermitian matrices.
\end{thm}

Following the strategy of Schuch and Verstraete~\cite{schuch09}, we will apply a sequence of perturbative gadgets to produce arbitrary interactions in the set $\{XX, XZ, ZX, ZZ\}$ from our allowed interaction terms. This will eventually give a Hamiltonian where all 2\nobreakdash-local interactions are of the form $\alpha XX + \beta YY + \gamma ZZ$ and take place on a sparse 2d square lattice, i.e.\ a square lattice with some vertices missing. If desired, heavily weighted local terms can then be used to effectively produce these gaps from a full lattice with equally weighted interactions across each connected pair of qubits (a similar idea was used in~\cite{schuch09}). To see this, observe that by Corollary \ref{cor:zerothorder}, using the gadget
\begin{center}
\begin{tikzpicture}[inner sep=0.5mm,xscale=6]
\node (a) at (0,0) [circle,fill=black,label=south:$a$] {};
\node (b) at (1,0) [circle,fill=black,label=south:$b$,label=north:$\Delta\proj{1}$] {};
\draw (a) to node[midway,above] {$\alpha XX + \beta YY + \gamma ZZ$} (b);
\end{tikzpicture}
\end{center}
for large enough $\Delta$ forces qubit $b$ into the state $\ket{0}$, decoupling it from qubit $a$. A side-effect of the gadget is the application of a new effective term $\alpha \bracket{0}{X}{0} X + \beta \bracket{0}{Y}{0}Y + \gamma \bracket{0}{Z}{0}Z$ to qubit $a$, but this can then be corrected using our freedom to apply arbitrary 1\nobreakdash-local terms. We can apply this gadget simultaneously to all qubits which we would like to remove from the lattice, correcting their neighbours afterwards.

The following lemma is a useful starting point.

\begin{replem}{lem:xxandzzlf}
Let $\alpha$ and $\beta$ be arbitrary fixed non-zero real numbers. {\sc $\{XX,ZZ\}$-Hamiltonian with local terms} is $\qma$-complete, even if all 2-qubit interactions are restricted to the edges of a 2d square lattice, the $XX$ terms all have weight $\alpha$, and the $ZZ$ terms all have weight $\beta$.
\end{replem}

\begin{proof}
We will show that we can generate effective $XX$ and $XZ$ terms, with arbitrary weights, ultimately using only $XX$ and $ZZ$ terms with fixed weights (producing $ZZ$ and $ZX$ terms can be done in the same way by relabelling Paulis). In both cases, we perform the reductive steps carefully: first to ensure that in either case we end up with the same number of new edges in the interaction graph (so it is a subgraph of a 2d square lattice), and second so that all of the edges are equally weighted.

Note that we can rescale all the interactions by an arbitrary fixed (polynomially large) coefficient without changing the complexity of the problem. In what follows, $\Delta$ and $\Delta'$ denote coefficients of this form. This allows us to assume in the proof that $|\alpha|$ and $|\beta|$ can be chosen arbitrarily large (but still fixed); we can always rescale them at the very end to the desired values.

First, using the $XX$ and $ZZ$ terms with fixed weights we can make an effective $XZ$ term with an arbitrary weight via the following gadget, where $\ket{\psi} = \cos\theta\ket{0} + \sin\theta\ket{1}$:
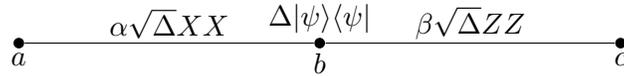
\begin{center}
\begin{tikzpicture}[inner sep=0.5mm,xscale=4]
\node (a) at (0,0) [circle,fill=black,label=south:$a$] {};
\node (b) at (1,0) [circle,fill=black,label=south:$b$,label=north:$\Delta\proj{\psi}$] {};
\node (c) at (2,0) [circle,fill=black,label=south:$c$] {};
\draw (a) to node[midway,above] {$\alpha \sqrt{\Delta} XX$} (b); \draw (b) to node[midway,above] {$\beta \sqrt{\Delta} ZZ$} (c);
\end{tikzpicture}
\end{center}
According to Lemma~\ref{lem:gadget}, up to local terms (and inverse-polynomially small corrections), we get the effective Hamiltonian
\[ \heff = \helse - 4\alpha \beta \cos\theta \sin\theta(\sin^2 \theta - \cos^2\theta) X_a Z_c = \helse + \alpha \beta \sin(4\theta) X_a Z_c. \]
Thus, if $|\alpha|$ and $|\beta|$ are large enough, by tuning $\theta$, we can make an arbitrarily weighted $XZ$ term, even though the $XX$ terms all have weight $\alpha$ and the $ZZ$ terms all have weight $\beta$. On the other hand, using a gadget of the form
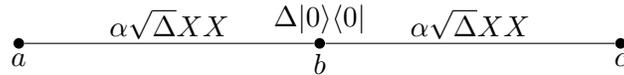
\begin{center}
\begin{tikzpicture}[inner sep=0.5mm,xscale=4]
\node (a) at (0,0) [circle,fill=black,label=south:$a$] {};
\node (b) at (1,0) [circle,fill=black,label=south:$b$,label=north:$\Delta\proj{0}$] {};
\node (c) at (2,0) [circle,fill=black,label=south:$c$] {};
\draw (a) to node[midway,above] {$\alpha \sqrt{\Delta}XX$} (b); \draw (b) to node[midway,above] {$\alpha \sqrt{\Delta} XX$} (c);
\end{tikzpicture}
\end{center}
we make the effective Hamiltonian
\[ \heff = \helse - \alpha^2 X_a X_c. \]
(This may seem unnecessary, as we already had access to the interaction $XX$ -- but we do it in order to make the number of reductive steps the same in all the cases we are considering, to produce a regular lattice.) Similarly, using a gadget with terms of the form $\beta \sqrt{\Delta} ZZ$ we make an effective Hamiltonian $\heff = \helse - \beta^2 Z_a Z_c$. We combine the effective interactions produced by these different gadgets in the following ways:
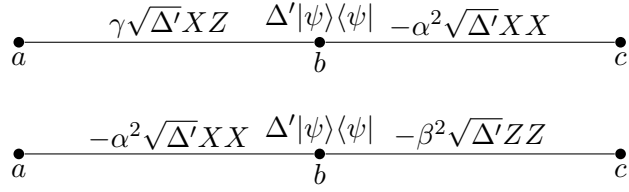
\begin{center}
\begin{tikzpicture}[inner sep=0.5mm,xscale=4]
\node (a) at (0,0) [circle,fill=black,label=south:$a$] {};
\node (b) at (1,0) [circle,fill=black,label=south:$b$,label=north:$\Delta'\proj{\psi}$] {};
\node (c) at (2,0) [circle,fill=black,label=south:$c$] {};
\draw (a) to node[midway,above] {$\gamma\sqrt{\Delta'} XZ$} (b); \draw (b) to node[midway,above] {$-\alpha^2\sqrt{\Delta'} XX$} (c);
\end{tikzpicture}
\end{center}
\begin{center}
\begin{tikzpicture}[inner sep=0.5mm,xscale=4]
\node (a) at (0,0) [circle,fill=black,label=south:$a$] {};
\node (b) at (1,0) [circle,fill=black,label=south:$b$,label=north:$\Delta'\proj{\psi}$] {};
\node (c) at (2,0) [circle,fill=black,label=south:$c$] {};
\draw (a) to node[midway,above] {$-\alpha^2\sqrt{\Delta'} XX$} (b); \draw (b) to node[midway,above] {$-\beta^2\sqrt{\Delta'} ZZ$} (c);
\end{tikzpicture}
\end{center}
where $\ket{\psi}$ is as before (but the angle $\theta$ can vary throughout) and $\gamma$ is an arbitrary weight produced by the initial gadget above. The first of these produces the effective Hamiltonian
\[ \heff = \helse - \alpha^2 \gamma \sin(4\theta)X_a X_c, \]
while the second produces the effective Hamiltonian
\[ \heff = \helse + \alpha^2\beta^2 \sin(4\theta) X_a Z_c. \]
As above, although we already have access to $XZ$, we do this to ensure the same number of reductive steps are used in both cases. Tuning $\gamma$ and $\theta$ appropriately allows us to produce effective arbitrarily weighted $XX$ and $XZ$ interactions; $ZZ$ interactions can be made in the same way as $XX$, by relabelling. The claim then follows from Theorem~\ref{thm:lattice}.
\end{proof}

\begin{replem}{lem:xxzzhamlf}
For any fixed $\gamma \neq 0$, {\sc $\{XX + \gamma ZZ\}$-Hamiltonian with local terms} is $\qma$-complete. This holds even if all 2-qubit interactions have the same weight and are restricted to the edges of a 2d square lattice.
\end{replem}

\begin{proof}
We use the following perturbative gadget:
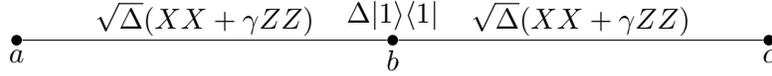
\begin{center}
\begin{tikzpicture}[inner sep=0.5mm,xscale=5]
\node (a) at (0,0) [circle,fill=black,label=south:$a$] {};
\node (b) at (1,0) [circle,fill=black,label=south:$b$,label=north:$\Delta\proj{1}$] {};
\node (c) at (2,0) [circle,fill=black,label=south:$c$] {};
\draw (a) to node[midway,above] {$\sqrt{\Delta}(XX + \gamma ZZ)$} (b); \draw (b) to node[midway,above] {$\sqrt{\Delta}(XX+\gamma ZZ)$} (c);
\end{tikzpicture}
\end{center}
According to Lemma~\ref{lem:gadget}, up to local terms and small corrections,
\[ \heff = \helse - 2 X_a X_c. \]
Thus, given access to terms of the form $XX + \gamma ZZ$, we can make $XX$ terms. Similarly, using the gadget
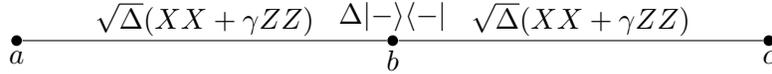
\begin{center}
\begin{tikzpicture}[inner sep=0.5mm,xscale=5]
\node (a) at (0,0) [circle,fill=black,label=south:$a$] {};
\node (b) at (1,0) [circle,fill=black,label=south:$b$,label=north:$\Delta\proj{-}$] {};
\node (c) at (2,0) [circle,fill=black,label=south:$c$] {};
\draw (a) to node[midway,above] {$\sqrt{\Delta}(XX + \gamma ZZ)$} (b); \draw (b) to node[midway,above] {$\sqrt{\Delta}(XX+\gamma ZZ)$} (c);
\end{tikzpicture}
\end{center}
we can produce an effective Hamiltonian
\[ \heff = \helse - 2 \gamma^2 Z_a Z_c. \]
As $\gamma \neq 0$, and Lemma~\ref{lem:xxandzzlf} holds for \emph{arbitrary} non-zero $\alpha$ and $\beta$, the claim follows.
\end{proof}

\begin{replem}{lem:xxyyzzhamlf}
For any fixed $\beta,\gamma \neq 0$, {\sc $\{XX + \beta YY + \gamma ZZ\}$-Hamiltonian with local terms} is $\qma$-complete. This holds even if all 2-qubit interactions have the same weight and are restricted to the edges of a 2d square lattice.
\end{replem}

\begin{proof}
Set $\ket{\psi} = \frac{1}{\sqrt{2}}(\ket{0} - i \ket{1})$ and use the following perturbative gadget:
\begin{center}
\begin{tikzpicture}[inner sep=0.5mm,xscale=6]
\node (a) at (0,0) [circle,fill=black,label=south:$a$] {};
\node (b) at (1,0) [circle,fill=black,label=south:$b$,label=north:$\Delta\proj{\psi}$] {};
\node (c) at (2,0) [circle,fill=black,label=south:$c$] {};
\draw (a) to node[midway,above] {$\sqrt{\Delta}(XX + \beta YY + \gamma ZZ)$} (b); \draw (b) to node[midway,above] {$\sqrt{\Delta}(XX+\beta YY +\gamma ZZ)$} (c);
\end{tikzpicture}
\end{center}
According to Lemma~\ref{lem:gadget}, up to local terms,
\[ \heff = \helse - 2 \left(XX + \gamma^2 ZZ\right). \]
We can therefore make terms of the form $XX + \gamma^2 ZZ$ for some $\gamma \neq 0$, so {\sc $\{XX + \gamma ZZ\}$-Hamiltonian with local terms} reduces to {\sc $\{XX + \beta YY + \gamma ZZ\}$-Hamiltonian with local terms}. By Lemma~\ref{lem:xxzzhamlf}, {\sc $\{XX + \beta YY + \gamma ZZ\}$-Hamiltonian with local terms} is $\qma$-complete.
\end{proof}

\begin{replem}{lem:xzskewlf}
{\sc $\{XZ - ZX\}$-Hamiltonian with local terms} is $\qma$-complete. This holds even if all 2-qubit interactions have the same weight and are restricted to the edges of a 2d square lattice.
\end{replem}

\begin{proof}
For this reduction, we use two perturbative gadgets. The first is
\begin{center}
\begin{tikzpicture}[inner sep=0.5mm,xscale=5]
\node (a) at (0,0) [circle,fill=black,label=south:$a$] {};
\node (b) at (1,0) [circle,fill=black,label=south:$b$,label=north:$\Delta\proj{1}$] {};
\node (c) at (2,0) [circle,fill=black,label=south:$c$] {};
\draw (a) to node[midway,above] {$\sqrt{\Delta}(XZ - ZX)$} (b); \draw (b) to node[midway,above] {$\sqrt{\Delta}(XZ-ZX)$} (c);
\end{tikzpicture}
\end{center}
By Lemma~\ref{lem:gadget}, up to local terms, this produces $\heff = \helse + 2 XX$. Similarly, if we set $\ket{+} = \frac{1}{\sqrt{2}}(\ket{0}+\ket{1})$ and use the gadget
\begin{center}
\begin{tikzpicture}[inner sep=0.5mm,xscale=5]
\node (a) at (0,0) [circle,fill=black,label=south:$a$] {};
\node (b) at (1,0) [circle,fill=black,label=south:$b$,label=north:$\Delta\proj{+}$] {};
\node (c) at (2,0) [circle,fill=black,label=south:$c$] {};
\draw (a) to node[midway,above] {$\sqrt{\Delta}(XZ - ZX)$} (b); \draw (b) to node[midway,above] {$\sqrt{\Delta}(XZ-ZX)$} (c);
\end{tikzpicture}
\end{center}
we can effectively produce $\heff = \helse + 2 ZZ$. The claim follows from Lemma~\ref{lem:xxandzzlf}.
\end{proof}


\section{The Lieb-Mattis model}
\label{app:liebmattis}

In this appendix we prove Lemma~\ref{lem:liebmattis}, which characterises the ground state of the Lieb-Mattis model. In order to do this, we first need to understand the Heisenberg model on the complete graph.

\begin{lem}
\label{lem:complete}
Let
\[ H_C := \sum_{i<j\in [n]} X_i X_j + Y_i Y_j + Z_i Z_j \]
be the Hamiltonian corresponding to the Heisenberg model on the complete graph on $n$ vertices. Then $H_C$ has eigenvalues 
$2s(s+1) - 3n/4$, where $s \in \{0,1,\dots,n/2\}$ if $n$ is even, and $s \in \{1/2,3/2,\dots,n/2\}$ if $n$ is odd.
\end{lem}

We prove Lemma~\ref{lem:complete} using the beautiful theory of spin. Lemma~\ref{lem:complete} is well-known in the condensed-matter literature, although often stated differently, and the proof technique is standard undergraduate quantum mechanics (see for example~\cite{ballentine98,sakurai11,hannabuss97}). However, it may not be familiar to computer scientists and we therefore present a completely self-contained (albeit also completely standard) proof.


For arbitrary $n$, define the matrices
\[ S^x := \frac{1}{2} \sum_{i=1}^n X_i,\;\; S^y := \frac{1}{2} \sum_{i=1}^n Y_i, \;\; S^z := \frac{1}{2} \sum_{i=1}^n Z_i,\;\; S^2 := (S^x)^2 + (S^y)^2 + (S^z)^2. \]
It can be checked that each of the operators $S^x$, $S^y$, $S^z$ commutes with $S^2$ (although not with each other), so we can find a basis consisting of simultaneous eigenvectors of $S^z$ and $S^2$. We will show the following.

\begin{lem}
\label{lem:spin}
Let the simultaneous eigenvalues of $S^2$, $S^z$ be indexed by pairs $(\lambda,m)$. Then:
\begin{enumerate}
\item $\lambda = s(s+1)$, where $s \in \{0,1,\dots,n/2\}$ if $n$ is even, and $s \in \{1/2,3/2,\dots,n/2\}$ if $n$ is odd;
\item $m$ is an integer multiple of $1/2$ satisfying $|m| \le s$.
\end{enumerate}
\end{lem}

Before proving this lemma, we observe that it implies Lemma~\ref{lem:complete} via
\[ S^2 = \frac{1}{4} \sum_{i,j\in [n]} X_i X_j + Y_i Y_j + Z_i Z_j = \frac{3n}{4} I + \frac{1}{2} \sum_{i<j\in [n]} X_i X_j + Y_i Y_j + Z_i Z_j = \frac{1}{2} H_C + \frac{3n}{4} I. \]

\begin{proof}
Define the ladder operators
\[ S^+ := S^x + i S^y = \sum_{i=1}^n \ket{1}\bra{0}_i,\;\; S^- := S^x - i S^y = \sum_{i=1}^n \ket{0}\bra{1}_i. \]
We summarise several useful identities involving these operators:
\begin{enumerate}[(i)]
\item $[S^x,S^y] = i S^z$, $[S^y,S^z] = i S^x$, $[S^z,S^x] = i S^y$;
\item $S^{\mp}S^{\pm} = S^2 - (S^z)^2 \mp S^z$;
\item $[S^z,S^{\pm}] = \pm S^{\pm}$.
\end{enumerate}
In these identities, $[\cdot,\cdot]$ is the commutator, $[A,B]:=AB-BA$. The proofs: for (i), we have
\[ [S^x,S^y] = S^x S^y - S^y S^x = \frac{1}{4} \sum_{i,j=1}^n X_i Y_j - Y_j X_i = \frac{1}{4} \sum_{i=1}^n X_i Y_i - Y_i X_i = i S^z, \]
and the other two cases are similar; for (ii), we have
\begin{align*}
S^{\mp}S^{\pm} &= (S^x \mp i S^y)(S^x \pm i S^y)\\
&= (S^x)^2 + (S^y)^2 \pm i S^x S^y \mp i S^y S^x\\
&=  (S^x)^2 + (S^y)^2 \pm i [S^x,S^y] \\
&=  S^2 - (S^z)^2 \mp S^z,
\end{align*}
and for (iii),
\[ [S^z,S^{\pm}] = [S^z,S^x \pm i S^y] = [S^z,S^x] \pm i [S^z,S^y] = i S^y \pm S^x = \pm S^{\pm}. \]
Let $\ket{\psi}$ be a common eigenvector of $S^2$ and $S^z$ such that
\[ S^2 \ket{\psi} = \lambda\ket{\psi},\;\; S^z \ket{\psi} = m\ket{\psi} \]
for some $\lambda$ and $m$. Then we claim that
\[ \| S^{\pm} \ket{\psi} \|^2 = (\lambda-m(m\pm 1))\|\ket{\psi}\|^2. \]
Indeed, by identity (ii) above we have
\[ \| S^{\pm} \ket{\psi} \|^2 = \bracket{\psi}{S^{\mp}S^{\pm}}{\psi} = \bracket{\psi}{\left(S^2 - (S^z)^2 \mp S^z\right)}{\psi} = (\lambda-m(m\pm 1))\|\ket{\psi}\|^2. \]
%
As $\| S^{\pm} \ket{\psi} \|$ is always non-negative, this implies that $\lambda-m(m\pm 1) \ge 0$, with equality if and only if $S^{\pm} \ket{\psi} = 0$. If $m \ge 0$, this gives $\lambda \ge m(m+1)$, while if $m \le 0$ this gives $\lambda \ge m(m-1)$. We have
\[ S^z S^{\pm} \ket{\psi} = ([S^z, S^{\pm}] + S^{\pm} S^z)\ket{\psi} = (\pm S^{\pm} + S^{\pm}S^z) \ket{\psi} = (m\pm 1)S^{\pm}\ket{\psi} \]
by identity (iii) so, for any $k$, $(S^{\pm})^k \ket{\psi}$ is an eigenvector of $S^z$ with eigenvalue $m \pm k$. This implies that there exist integers $p,q \ge 0$ such that $(S^+)^{p+1}\ket{\psi} = 0$, but $(S^+)^p\ket{\psi} \neq 0$, and $(S^-)^{q+1}\ket{\psi} = 0$, but $(S^-)^q\ket{\psi} \neq 0$. By the analysis of the case of equality, we have $\lambda = (m+p)(m+p+1)$ and also $\lambda = (m-q)(m - q - 1)$. So
\[ (m+p)(m+p+1) = (m-q)(m-q-1), \]
implying $m = (q-p)/2$, i.e.\ is an integer multiple of $1/2$. If we set $s = (q+p)/2 = m+p$, then $\lambda = s(s+1)$, where $s$ is a positive integer multiple of $1/2$. By the inequalities relating $\lambda$ and $m$, this implies $|m| \le s$ (claim (2) of the lemma). Also, $p \le n/2-m$, because eigenvectors of $S^z$ with eigenvalue $m$ correspond to states with Hamming weight $n/2+m$ in the computational basis, which are all zeroed by at most $n/2-m$ applications of $S^+$. This implies $s \le n/2$, proving claim (1) of the lemma. Observe that we can construct eigenvectors of $S^z$ and $S^2$ with different values of $m$ by applying $S^{\pm}$ to $\ket{\psi}$ (note that this does not affect $\lambda$!) to obtain anything in the range $\{m-q,\dots,m+p\} = \{-s,\dots,s\}$. We know that there must exist some starting vector $\ket{\psi}$ from general arguments.
\end{proof}

We now use Lemma~\ref{lem:complete} to study the Lieb-Mattis model -- the Heisenberg model on a complete bipartite graph of size $2n$. Define the symmetric, Hamming-weight $k$ state
\[ \ket{\psi^n_k} := \frac{1}{\sqrt{\binom{n}{k}}} \sum_{x \in \{0,1\}^n, |x|=k} \ket{x}. \]
Then the following lemma combines results stated elsewhere in the literature (particularly~\cite{lieb62}, but also e.g.~\cite{vidal07}).

\begin{replem}{lem:liebmattis}
Write
\[ H_{LM} = \sum_{i=1}^n \sum_{j=n+1}^{2n} X_i X_j + Y_i Y_j + Z_i Z_j. \]
Then the ground state of $H_{LM}$ is unique and given by
\[ \ket{\phi_{LM}} := \frac{1}{\sqrt{n+1}} \sum_{k=0}^n (-1)^k \ket{\psi^n_k}\ket{\psi^n_{n-k}}. \]
For $i$ and $j$ such that $1 \le i,j \le n$ or $n+1 \le i,j \le 2n$, $\bracket{\phi_{LM}}{F_{ij}}{\phi_{LM}} = 1$. Otherwise, $\bracket{\phi_{LM}}{F_{ij}}{\phi_{LM}} = -2/n$.
\end{replem}

\begin{proof}
We first verify that $\ket{\phi_{LM}}$ is indeed an eigenvector of $H$. We use
\[ H_{LM} =  \sum_{i=1}^n \sum_{j=n+1}^{2n} (2F-I)_{ij} = -n^2 I + 2\sum_{i=1}^n \sum_{j=n+1}^{2n} F_{ij} =: -n^2 I + 2 H_F \]
and compute
\[ H_F \ket{\phi_{LM}} = \frac{1}{\sqrt{n+1}} \sum_{i,j=1}^n \sum_{k=0}^n (-1)^k F_{ij} \ket{\psi^n_k}\ket{\psi^n_{n-k}}. \]
After some tedious algebra, we get
\[ H_F \ket{\phi_{LM}} = -n \ket{\phi_{LM}}, \]
so $\ket{\phi_{LM}}$ is an eigenvector of $H_{LM}$ with eigenvalue $-n(n+2)$. We now give a matching lower bound, thus proving that every ground state has energy at least this large.
Define the following operators, where for conciseness we write $M := XX + YY + ZZ$:
\[ S = \sum_{i<j=1}^{2n} M_{ij},\;\;\;\; S_A = \sum_{i<j=1}^{n} M_{ij},\;\;\;\; S_B = \sum_{i< j=n+1}^{2n} M_{ij},\;\;\;\; S^z = \sum_{i=1}^{2n} Z_i. \]
Then the set of operators $\{H_{LM}, S, S_A, S_B, S^z\}$ commutes pairwise, and we have
\[ H_{LM} = S - S_A - S_B. \]
Because of these two facts, and as (by Lemma~\ref{lem:complete}) $S$ has eigenvalues in the set $2s(s+1) - 3n/2$, where $s \in \{0,1,\dots,n\}$,
and $S_A$, $S_B$ have eigenvalues in the set $2s(s+1) - 3n/4$, where $s \in \{0,1,\dots,n/2\}$ if $n$ is even, and $s \in \{1/2,3/2,\dots,n/2\}$ if $n$ is odd, the eigenvalues of $H_{LM}$ must be in the set
%
\[ \{ 2(s(s+1) - t(t+1) - u(u+1)): s \in \{0,1,\dots,n\}, t,u \in \{0,1,\dots,n/2\}\} \]
if $n$ is even, and
\[ \{ 2(s(s+1) - t(t+1) - u(u+1)): s \in \{0,1,\dots,n\}, t,u \in \{1/2,3/2,\dots,n/2\}\} \]
if $n$ is odd. In either case, this is clearly minimised by taking $s=0$, $t = u = n/2$, implying a lower bound on the smallest eigenvalue of $-n(n+2)$, which we have already seen can be achieved.


We still need to prove uniqueness of this ground state. By commutativity, we can find a common set of eigenvectors of $H_{LM}$ and $S^z$. Eigenvectors of $S^z$ are given by states $\ket{\psi} = \sum_{x \in \{0,1\}^{2n}} \alpha_x \ket{x}$ such that, for all $x$ such that $\alpha_x \neq 0$, $|x|=n-m$ for some fixed $m$ such that $-n \le m \le n$. (``$\ket{\psi}$ has total spin $m$ in the Z direction''.) $\ket{\phi_{LM}}$ is of this form with $m=0$; we now show that there are no other ground states with $m=0$.
%
 If we conjugate $H_{LM}$ by $Z$ matrices on the first $n$ qubits, we get
\[ H'_{LM} = \sum_{i=1}^n \sum_{j=n+1}^{2n} -X_i X_j - Y_i Y_j + Z_i Z_j = -2 \sum_{i=1}^n \sum_{j=n+1}^{2n} \left( (\ket{01}+\ket{10})(\bra{01}+\bra{10})_{ij} - I\right). \]
Ignoring the identity terms and rescaling, this is the negation of a matrix whose entries are all non-negative, and which is irreducible when restricted to a subspace of vectors with fixed Hamming weight. That is, thinking of $-H'_{LM}$ as the adjacency matrix of an undirected graph, there is a path from any vector of weight $k$ to any other vector of weight $k$, for all $k$. By the Perron-Frobenius theorem, this implies that the principal eigenvector of the matrix equal to $-H'_{LM}$, restricted to this subspace, has strictly positive entries everywhere. This in turn implies that there is only one such vector on this subspace (as two vectors of this form could not be orthogonal). Therefore, $H_{LM}$ can only have at most one ground state on each such subspace.

Finally, we need to show that if $m\neq 0$, there are no other ground states. It suffices to show that such states cannot be ground states of $S$. But this follows from Lemma~\ref{lem:spin}, because eigenvalues $\lambda$ of $S^2$ satisfy $\lambda = s(s+1)$, where $s \ge |m|$.

For the second part, it is immediate that $\ket{\phi_{LM}}$ is left unchanged by a flip of two indices which both belong either to the first or second block, so $\bracket{\phi_{LM}}{F_{ij}}{\phi_{LM}} = 1$. 
%
 Further tedious algebra suffices to compute $\bracket{\phi_{LM}}{F_{ij}}{\phi_{LM}} = -2/n$.
\end{proof}


\bibliographystyle{plain}
\bibliography{lham}

\end{document}